\documentclass[letterpaper,onecolumn,american,showpacs,pr,aps,superscriptaddress,eqsecnum,nofootinbib]{revtex4}
 
\usepackage{graphicx}

\usepackage{setspace}

\usepackage{amsfonts}
\usepackage{amsmath}
\usepackage{amssymb}
\usepackage{graphicx}%
\usepackage{color}
\providecommand{\U}[1]{\protect\rule{.1in}{.1in}}
\newtheorem{theorem}{Theorem}

\newtheorem{corollary}[theorem]{Corollary}

\newtheorem{lemma}[theorem]{Lemma}

\newtheorem{proposition}[theorem]{Proposition}
\newtheorem{remark}[theorem]{Remark}

\newenvironment{proof}[1][Proof]{\noindent\textbf{#1.} }{\ \rule{0.5em}{0.5em}}

\definecolor{nblue}{rgb}{0.2,0.2,0.7}
\definecolor{ngreen}{rgb}{0.2,0.6,0.2}
\definecolor{nred}{rgb}{0.7,0.2,0.2}
\definecolor{nblack}{rgb}{0,0,0}

\begin{document}

\title{A generalization of Schur-Weyl duality with applications in quantum estimation}
\author{Iman Marvian}
\affiliation{Perimeter Institute for Theoretical Physics,  Waterloo,
Ontario, Canada N2L 2Y5}
\affiliation{Institute for Quantum Computing, University of Waterloo, Waterloo, Ontario, Canada N2L 3G1}
\affiliation{Department of Physics and Astronomy, Center for Quantum Information Science and Technology, University of Southern California, Los Angeles, CA 90089}

\author{Robert W. Spekkens}
\affiliation{Perimeter Institute for Theoretical Physics,  Waterloo,
Ontario, Canada N2L 2Y5}

\date{\today}
\begin{abstract}
Schur-Weyl duality is a  powerful tool in representation theory which has many applications to quantum information theory.  We provide a generalization of this duality and demonstrate some of its applications.  In particular, we use it to develop a general framework for the study of a family of quantum estimation problems wherein one is given $n$ copies of an unknown quantum state according to some prior and the goal is to estimate certain parameters of the given state. In particular, we are interested to know whether collective measurements are useful and if so to find an upper bound on the amount of entanglement which is required to achieve the optimal estimation.  In the case of pure states, we show that commutativity of the set of observables that define the estimation problem implies the sufficiency of unentangled measurements.
\end{abstract}
\maketitle
\setstretch{1.2}

\tableofcontents

\section{Introduction}

Schur-Weyl duality is a duality between two subgroups of the general linear group on $(\mathbb{C}^d)^{\otimes n}$: the \emph{collective action} of the unitary group $\text{U}(d)$, and the \emph{canonical representation} of the group $\mathcal{S}_n$ of permutations of the $n$ systems (See section \ref{Chapt-per-II} for precise definitions and the statement of the duality).  It asserts that there is a one-to-one map between the irreducible representations of the two groups, and that their product is multiplicity-free.  Alternatively, one can characterize the duality as the fact that the complex algebra spanned by one of these groups is the commutant of the one spanned by the other. The generalization we derive here is also between two subgroups of the general linear group on $(\mathbb{C}^d)^{\otimes n}$.  One is the collective action of a subgroup $G$ of $\text{U}(d)$, where $G$ has a particular property, namely, that it is equal to the centralizer of its centralizer in $\text{U}(d)$. We call such a group a \emph{gauge} group (for reasons that will be explained shortly).  The other is the group closure of the local action of $G'$ (the centralizer of $G$) and the canonical representation of the permutation group $\mathcal{S}_n$.  Schur-Weyl duality is included as the special case where $G=\text{U}(d)$.

Just as Schur-Weyl duality has many applications to quantum information theory and quantum algorithms (see \cite{Goodman-Wallach} and \cite{Aram} for a review), so too does this generalization.  This article will explore some of these applications.  

One such application is to finding noiseless subsystems (this is considered in Sec.~\ref{sec:noiseless}). However, most of the applications will rely on a particular consequence which connects global symmetries with local symmetries, considered in Sec.~\ref{sec:Promote}.

For $M$ an arbitrary operator on $(\mathbb{C}^d)^{\otimes n}$, we say that $M$ has \emph{global symmetry} with respect to the subgroup $H$ of $\text{U}(d)$ if it is invariant under  the collective action of $H$, i.e.,
\begin{equation}\label{global}
\forall V\in H:\  V^{\otimes n} M {V^{\dag}}^{\otimes n} =  M,
\end{equation}
and we say that $M$ has \emph{local symmetry} with respect to $H$ if  it is invariant under the local action of $H$, i.e.,
\begin{align}\label{local}
 \forall V\in H\ \ \text{and}\ \  \forall k: 0\le k\le n-1,  
\ \ \ (I^{\otimes k} \otimes V \otimes I^{\otimes (n-k-1)} ) M  (I^{\otimes k} \otimes V^{\dag} \otimes I^{\otimes (n-k-1)} )=  M
\end{align}
Note that any operator which has local symmetry with respect to $H$ automatically also has global symmetry with respect to $H$ but the converse implication does not necessarily hold. Indeed, generally the condition of local symmetry is much stronger than that of global symmetry.

The duality implies that within the totally symmetric and the totally antisymmetric subspaces of $(\mathbb{C}^d)^{\otimes n}$, the collective action of a gauge group $G$ is dual to the collective action of $G'$, its centralizer in $U(d)$.  This in turn implies that if an operator is confined to the totally symmetric or totally antisymmetric subspace and has \emph{global} symmetry with respect to the gauge group $G$, then it must also have \emph{local} symmetry  with respect to $G$. In other words, our generalization of Schur-Weyl duality allows in some circumstances for a global symmetry to be promoted to a local symmetry. 

The main application we consider is the problem of how to best estimate parameters describing a quantum state given multiple copies of the state (this is considered in Sec.~\ref{sec:estimation}). 
 The parameters might consist of expectation values of some observables, or they might encode a decision about the state, such as whether a given expectation value is positive or not.  In particular, we seek to determine under what circumstances it is sufficient to do independent measurements on each copy and in what circumstances more complicated measurements, for instance, using entanglement, are required.
(As it turns out, there are many circumstances wherein entangled measurements do help.)


A very simple example of such a multi-copy estimation problem is the one considered by Hayashi \emph{et al.} \cite{HHH}. A pure state is chosen uniformly according to the Haar measure, and $n$ copies of the state are prepared.  The goal is to estimate the expectation value of an observable $A$ for the state.  Hayashi \emph{et al.} have shown that for a squared-error figure of merit, the optimal estimation scheme is to simply measure the observable $A$ separately on each system.  Our generalization of Schur-Weyl duality can be used to provide a very elementary proof of this result.  It can also be used to simplify the solution of estimation problems that are much more complicated, as we shall show.

The reason we can make use of our generalization of Schur-Weyl duality is that a multi-copy estimation problem can be shown to naturally have a global symmetry for a gauge group.
Measurements with global symmetry relative to a gauge group are described by POVMs all the elements of which have this symmetry.  
In this case, the duality tells us that the global symmetry can be promoted to a local symmetry, so it suffices to consider measurements on the $n$ copies that have \emph{local symmetry} relative to the gauge group.  

We now explain how this promotion of global symmetries to local symmetries immediately leads to the solution of the multi-copy estimation problem considered by Hayashi \emph{et al.}.  The prior is uniform over pure states and the squared error figure of merit only relies on the observable ${A}$ that one is trying to estimate. Consequently, the description of this problem has symmetry with respect to the group of all unitaries which commute with ${A}$ and  it follows that, on the multi-copy system, it suffices to perform measurements that have global symmetry with respect to this group.  But this group is a gauge group, i.e. it is equal to the centralizer of its centralizer in the unitary group, and so by our result, one can promote this global symmetry to a local symmetry.   Finally, noting that all measurement operators that are invariant under the local action of the gauge group must be in the algebra generated by the set
$$\{I^{\otimes k}\otimes A\otimes I^{\otimes (n-k-1)}:\   0\le k\le n-1 \} \cup \{I^{\otimes n}\},$$
it follows that one can simply measure ${A}$ on each copy individually and do classical processing on the outcomes to achieve the optimal estimation.  We also immediately see that even if the figure of merit is not the squared error and the prior over pure states is not uniform, as long as these depend only on ${A}$, then an individual measurement on each copy continues to be optimal.

In more complicated examples, wherein there may be many observables to be estimated, a non-uniform prior over pure states and an arbitrary figure of merit, as long as the problem still has some gauge symmetry we can exploit our result to infer that the optimal measurement on the $n$ copies should have local symmetry with respect to the gauge group.

In particular, if the commutant of the gauge group is a commutative algebra, then it suffices to implement independent measurements of the generator of this algebra on each system.  This occurs if the problem is to estimate a set of commuting observables, and the figure of merit and the prior over pure states can be entirely described in terms of this set of observables (hence, no entanglement is needed). Furthermore, even if the commutant of the problem's gauge group is not a commutative algebra, so that independent measurements are \emph{not} sufficient, nonetheless local symmetry is still a stronger constraint than global symmetry and consequently our result can lead to a bound on how much interaction between the systems is required to achieve the optimal measurement.

We also demonstrate several other generalizations of the basic multi-copy estimation problem -- to cases which include mixed states and to cases where the channel between the source and the estimator can be noisy --
 such that our results still have nontrivial consequences for the optimal measurement.
 
Finally, we demonstrate that our result has applications for estimation problems where the estimator gets only a single copy of the system of interest, and the distinction between global and local symmetries is relative to the partitioning of the system of interest into its components. In particular, we obtain strong constraints on the optimal measurement in the case of a system with \emph{two} components because the permutation group on two systems has only irreducible representations over the symmetric and antisymmetric subspaces and our duality permits an inference from global symmetry to local symmetry within the symmetric and antisymmetric subspaces.  This is considered in Sec.~\ref{sec:bipartite}.

Given that the class of estimation problems for which our results apply is very large, they represent a dramatic expansion, relative to previously known results, in the scope of problems for which we can easily determine the optimal measurement.
Furthermore, in previous results where independent measurements on each copy were shown to be optimal, such as Ref.~\cite{HHH}, the reasoning was rather \emph{ad hoc}.  It was not clear what feature of the estimation problem implied the sufficiency of such measurements.  By contrast, our approach follows a clear methodology -- we are determining the consequences of the gauge symmetries of the estimation problem.  Our results establish a sufficient condition for the optimality of independent measurements, i.e., the lack of any need for adaptive or entangled measurements. It is that the set of single-copy observables that are needed to define the estimation problem form a \emph{commutative} set.  In a slogan, \emph{the commutativity of the observables defining the estimation problem imply the adequacy of independent measurements.}


\section{Preliminaries}\label{Chapt-per-II}

\subsection{Commutant and Centralizer}
For a complex vector space $\mathcal{V}$, define $\text{End}(\mathcal{V})$  to be the set of linear maps from $\mathcal{V}$ to itself (endomorphism). This set has a natural algebraic structure and is called the \emph{full matrix algebra} over $\mathcal{V}$. Any matrix algebra defined on $\mathcal{V}$  is a subalgebra of $\text{End}(\mathcal{V})$. Here, we only consider finite dimensional  vector spaces.

For any vector space $\mathcal{V}$, and any set $\{A_{i}\in \text{End}(\mathcal{V})\}$ we call the set of all operators in $\text{End}(\mathcal{V})$ which commute with $\{A_{i}\}$ the \emph{commutant} of $\{A_{i}\}$ and denote it by $\text{Comm}\{A_{i}\}$. Note that for any arbitrary set $\{A_{i}\in \text{End}(\mathcal{V})\}$, its commutant, i.e. $\text{Comm}\{A_{i}\}$, is an algebra.

Let $\{A_{i}\in \text{End}(\mathcal{V})\}$ be a set of Hermitian operators, i.e. $A_{i}=A^{\dag}_{i}$. Then  it holds that

\begin{equation}\label{com-com}
\text{Comm}\{\text{Comm}\{A_{i}\}\}=\text{Alg}\{A_{i},\mathbb{I}\}
\end{equation}
where by $\text{Alg}\{A_{i},\mathbb{I}\}$ we mean the complex matrix algebra generated by the set $\{A_{i}\}$ and $\mathbb{I}$ (the identity operator on $\mathcal{V}$).  Any such complex matrix algebra which includes the identity operator and is closed under  adjoint ($\dag$) is called a \emph{finite dimensional von Neumann algebra}.  Note that Eq.(\ref{com-com}) means that for any finite dimensional von Neumann algebra $\mathcal{A}$,
\begin{equation} \label{com-com2}
\text{Comm}\{\text{Comm}\{\mathcal{A}\}\}=\mathcal{A}
\end{equation}
which is the defining property of these algebras.  In this paper we only use this type of algebra and whenever we refer to an object as an algebra we mean a  {finite dimensional von Neumann algebra}. Note that for any subgroup  $H$ of the unitary group, the algebra spanned by $H$, $\text{Alg}\{H\}$, is a von Neumann algebra.

A finite dimensional von Neumann algebra, as a finite dimensional matrix $\text{C}^{\ast}$-algebra, has a unique decomposition up to unitary equivalence of the form,
\begin{equation}\label{decom-algeb}
\mathcal{A}\cong \bigoplus_{J} \left(\mathcal{M}_{m_{J}}\otimes \mathbb{I}_{n_{J}} \right),
\end{equation}
where $\mathcal{M}_{m_{J}}$ is the full matrix algebra $\text{End}(\mathbb{C}^{m_{J}})$ and $\mathbb{I}_{n_{J}}$ is the identity on $\mathbb{C}^{n_{J}}$.   A von Neumann algebra by definition includes identity. Therefore for these algebras $\sum_{J} m_{J}n_{J}$ is equal to the dimension of the vector space.

   For two algebras $\mathcal{A}_{1}\subseteq \text{End}(\mathcal{V}_{1})$ and $\mathcal{A}_{2}\subseteq \text{End}(\mathcal{V}_{2})$ it holds that
\begin{equation}\label{com-tensor-product}
\text{Comm}\{\mathcal{A}_{1}\otimes \mathcal{A}_{2} \}=\text{Comm}\{\mathcal{A}_{1}\}\otimes \text{Comm}\{\mathcal{A}_{2} \}
\end{equation}
this is called \emph{the commutation theorem for tensor products}.

In this paper we will use the notion of \emph{centralizer} in a different way than \emph{commutant}. By the \emph{centralizer} of a subgroup $H_{0}$ in group $H$ we mean the set of all elements of group $H$ which commute with all elements of the subgroup $H_{0}$. We denote the centralizer of $H_{0}$  by $H'_{0}$. Note that the centralizer of any subgroup of a group is also a subgroup of that group.

Let $H$ be a subgroup of $\text{U}(d)$ and $H'$ be its centralizer in this group. Then it holds that
 \begin{equation}\label{comm-alg}
 \text{Comm}\{H\}=\text{Alg}\{H' \}.
\end{equation}

\subsection{Dual reductive pairs and Schur-Weyl duality}

Let $H_{1}$ and $H_{2}$ be two groups of unitaries acting on the complex vector space $\mathcal{V}$ and assume that they commute with each other, that is, $H_{1}$ and $H_{2}$ are each within one another's centralizer in the group of all unitaries on $\mathcal{V}$. Then,  under the action of $H_{1}$ and $H_{2}$, the space $\mathcal{V}$ decomposes as follows
\begin{equation} \label{decom}
\mathcal{V}\cong\sum_{\mu,\nu} \mathcal{M}_{\mu}\otimes \mathcal{N}_{\nu}\otimes  \mathbb{C}^{m_{\mu,\nu}}
\end{equation}
where $H_{1}$ and $H_{2}$  act irreducibly on $ \mathcal{M}_{\mu}$ and $ \mathcal{N}_{\mu}$ respectively, where $\mu$ and $\nu$ label distinct irreducible representations (irreps) of $H_{1}$ and $H_{2}$ respectively and where $m_{\mu,\nu}$ is the multiplicity of the pair of irreps $\mu,\nu$.    Then for some specific commuting groups the following equivalent properties hold \cite{Goodman-Wallach, Aram}.

\begin{proposition} \label{prop-doub-comm}
Let $H_{1}$,$H_{2}$ be two  groups acting on $\mathcal{V}$. Then the following are equivalent
\begin{enumerate}
\item The complex algebra spanned by $H_{1}$ is the commutant of the complex algebra spanned by $H_{2}$ in $\text{End}(\mathcal{V})$ and vice versa.
\item In the decomposition  \ref{decom} each $m_{\mu,\nu}$   is either 0 or 1 and at most one $m_{\mu,\nu}$ is nonzero for each $\mu$ and each $\nu$. \label{cond2}
\end{enumerate}
Any two groups with these properties are called a \emph{dual reductive pair} of subgroups of $\text{GL}(\mathcal{V})$, the general linear group on $\mathcal{V}$.
\end{proposition}

Note that, using the notation we have introduced before, the first statement can be written as $\text{Alg}\{H_{1}\}=\text{Comm}\{H_{2}\}$ and by virtue of Eq.(\ref{com-com2}) this equation is  equivalent to $\text{Alg}\{H_{2}\}=\text{Comm}\{H_{1}\}$.

Consider the following representation of the unitary group $\text{U}(d)$  on $(\mathbb{C}^{d})^{\otimes n}$:
\begin{equation}
\forall V\in \text{U}(d) :\ \ \textbf{Q}(V) |i_{1}\rangle\otimes\cdots\otimes |i_{n}\rangle =V|i_{1}\rangle\otimes\cdots  \otimes V|i_{n}\rangle
\end{equation}
For a subgroup $H$ of $\text{U}(d)$ we denote the group $\{ \textbf{Q}(V) : V\in H\}$ by $\textbf{Q}(H)$ and we call it the \emph{collective action} of $H$ on $(\mathbb{C}^{d})^{\otimes n}$.  Consider also the \emph{canonical representation} of the symmetric group of degree $n$, $\mathcal{S}_{n}$, on $(\mathbb{C}^{d})^{\otimes n}$
\begin{equation}
\forall s\in \mathcal{S}_{n}:\ \ \textbf{P}(s) |i_{1}\rangle\otimes\cdots\otimes |i_{n}\rangle =|i_{s^{-1}(1)}\rangle\otimes\cdots  \otimes |i_{s^{-1}(n)}\rangle
\end{equation}
We denote the group $\{\textbf{P}( s):s\in \mathcal{S}_{n}\}$  by $\textbf{P}( \mathcal{S}_{n})$.
Then Schur-Weyl duality states that 

\begin{theorem} \label{Thm-Schur-Weyl duality}  (\textbf{Schur-Weyl duality}) The  following two algebras are commutants of one another in $\text{End}((\mathbb{C}^d)^{\otimes n})$
\begin{enumerate}
\item $\text{Alg}\{\textbf{Q}(\text{U}(d))\}$, the complex algebra spanned  by $\textbf{Q}(\text{U}(d))$.
\item $\text{Alg}\left\{\textbf{P}(\mathcal{S}_{n})\right\}$, the complex  algebra spanned by   $ \textbf{P}(\mathcal{S}_{n})$.
\end{enumerate}
In other words, the subgroups $\textbf{Q}(\text{U}(d))$ and  $\textbf{P}(\mathcal{S}_{n})$ are dual reductive pairs in $\text{GL}((\mathbb{C}^{d})^{\otimes n})$.
\end{theorem}
Using our notation, Schur-Weyl duality can be expressed as
$\text{Comm}\{\textbf{Q}(\text{U}(d))\} = \textrm{Alg} \{ \textbf{P}(\mathcal{S}_{n})\}$, 
or equivalently as
$\text{Alg}\{\textbf{Q}(\text{U}(d))\} = \text{Comm}\{ \textbf{P}(\mathcal{S}_{n})\}.$

This theorem together with the proposition \ref{prop-doub-comm} implies that there is a one-to-one correspondence between the irreps of the group $\text{U}(d)$ which show up in representation $\textbf{Q}(\text{U}(d))$ and the irreps of the group $\mathcal{S}_{n}$ which show up in  representation $\textbf{P}(\mathcal{S}_{n})$. Furthermore, the theorem implies that the action of $\textbf{Q}(\text{U}(d))\times \textbf{P}(\mathcal{S}_{n})$ is multiplicity-free on $(\mathbb{C}^d)^{\otimes n}$.

In the following section,  we present a generalization of Schur-Weyl duality for  the case of \emph{gauge} subgroups of $\text{U}(\text{d})$.


\section{A Generalization of Schur-Weyl duality} \label{chap-gen-Schur-Weyl}

\subsection{Gauge groups and their characterizations}\label{sec:charac}

For any subgroup $G$ of  $\text{U}(d)$ let $G'$ denote the centralizer of $G$ in $\text{U}(d)$, i.e. the set of all elements of $\text{U}(d)$ which commute with all elements of $G$. Also denote the centralizer of the centralizer of $G$ by $G''\equiv (G')'$. Then in general $G\subseteq G''$. We call a unitary group $G$  a \emph{gauge group} if $G=G''$. The fact that in any arbitrary group and for any arbitrary subgroup $H$, $H\subseteq H''$ implies that  $((H')')'=H'$. So  for arbitrary subgroup $H$ of  $\text{U}(d)$,  its centralizer  $H'$ is a gauge group. 

Equivalently, one can think of a gauge group as the set of all unitaries in $\text{End}(\mathbb{C}^{d})$ which commute with a von Neumann algebra $\mathcal{A}\subseteq \text{End}(\mathbb{C}^{d})$. This is true because for any subgroup $G$ of $\text{U}(d)$, $G''$ is equal to all the unitaries which commute with $G'$ or equivalently all the unitaries which commute with $\text{Alg}\{G'\}$ (which is a von Neumann algebra). So if $G=G''$ then $G$ is equal to the set of all unitaries which commute with an algebra, namely $\text{Alg}\{G'\}$. On the other hand, if $G$ is equal to the set of all unitaries in $\text{End}(\mathbb{C}^{d})$ which commute with an algebra $\mathcal{A}\subseteq \text{End}(\mathbb{C}^{d})$ then $G'$ is equal to the set of all the unitaries in the algebra $\mathcal{A}$ and so is a basis for this algebra. Since $G''$ is equal to the set of all the unitaries which commute with $G'$, and $G'$ is a basis for $\mathcal{A}$, then $G''$ is equal to the set of all unitaries which commute with the algebra $\mathcal{A}$ and so is equal to $G$.  Therefore, these two definitions of gauge group are equivalent.

This discussion implies that one way to specify a gauge group is to specify the von Neumann algebra of operators which commute with the gauge group, for instance by specifying the generators of that algebra.  We call the gauge group formed by all unitaries which commute with a von Neumann algebra  $\mathcal{A}$ \emph{the gauge group of}  $\mathcal{A}$  and denote it by $G_{\mathcal{A}}$.    Note that if $G_{\mathcal{A}}$ is the gauge group  of  $\mathcal{A}$  then it holds that

\begin{equation}\label{com-gauge}
\text{Comm}\{G_{\mathcal{A}}\}=\text{Alg}\{G'_{\mathcal{A}}\}=\mathcal{A}.
\end{equation}

Using this together with the commutation theorem for tensor product, Eq.(\ref{com-tensor-product}) and Eq.(\ref{comm-alg}), we find
\begin{equation} \label{comm-tens}
\text{Comm}\{G^{\times n}_{\mathcal{A}}\}=\text{Alg}\{(G'_{\mathcal{A}})^{\times n}\}=\mathcal{A}^{\otimes n}
\end{equation}

Also note that Eq.(\ref{com-gauge})  implies that any von Neumann algebra can be uniquely specified by  its gauge group.

Now, based on  this observation that any gauge group can be thought as the set of unitaries commuting with a von Neumann  algebra, characterizing  the set of all gauge groups  is equivalent to characterizing all von Neumann algebras, which is done by Eq.(\ref{decom-algeb}). This decomposition implies that $G_{\mathcal{A}}$, the gauge group of $\mathcal{A}$, has a unique  decomposition up to unitary equivalence of the form
\begin{equation}
G_\mathcal{A}\cong \bigoplus_{J} \left(\mathbb{I}_{m_{J}}\otimes \text{U}(n_{J}) \right)
\end{equation}
where $\mathbb{I}_{m_{J}}$ is the identity on $\mathbb{C}^{m_{J}}$ and $\sum_{J} n_{J}m_{J}=d$. In other words, for any set of integers $0 \le n_{1}\le \cdots \le n_{d}\le d$ there is a  gauge group acting on $\mathbb{C}^{d}$ which is isomorphic to   $\text{U}(n_1)\times \cdots \times\text{U}(n_d)$   iff there is a set of positive integers $ 1\le m_{1},\cdots  ,m_{d}\le d$ such that $\sum_{i=1}^{d} n_{i}m_{i}=d$ (Here, we use the convention that $U(0)$ is the trivial group which includes only one element.). 
In particular, for  any vector space $\mathbb{C}^{d}$, there are gauge groups isomorphic to  ${\text{U}(1)}^{\times d}$ and $\text{U}(d)$. These gauge groups can be   respectively thought of as the gauge group of the algebra of all diagonal matrices in some orthonormal basis and the algebra generated by the identity matrix.

For instance, in the case of $d=2$, the set of all gauge groups can be classified into the following three types:    i) $n_{1}=0,n_{2}=1$ which corresponds to the group $\{e^{i\theta} \mathbb{I}:\theta\in(0,2\pi]\}$ where $\mathbb{I}$ is the identity operator , ii) $n_{1}=0,n_{2}=2$ which corresponds to the group $\text{U}(2)$ iii)  $n_{1}=1,n_{2}=1$ which corresponds to the group
$$\{e^{i\theta_{0}}|0\rangle\langle 0|+e^{i\theta_{1}}|1\rangle\langle 1|:\theta_{0},\theta_{1}\in (0,2\pi]\}$$
for any arbitrary orthonormal basis $\{|0\rangle,|1\rangle\}$.

Note that this characterization implies that any non-trivial gauge group is a unimodular Lie group, i.e. its left invariant measure is equal to the right invariant measure (up to a constant) and so it has a unique invariant measure.

Throughout the rest of this paper  we will extensively use the uniform twirling over subgroups of the unitary group with respect to their unique (normalized) Haar measure. For subgroup $H$ of $\text{U}(d)$ we denote this uniform twirling by  
\begin{equation}\label{Def:twirling}
\mathcal{T}_{H}(\cdot)\equiv\int_{H} d\mu(V)\ V(\cdot)V^{\dag}
\end{equation}
where $d\mu$ is the normalized Haar measure of $H$. Since $d\mu$ is the uniform measure any operator in the image of $\mathcal{T}_{H}$ commutes with $H$. Therefore if $G_{\mathcal{A}}$ is the gauge group of a von Neumann  algebra $\mathcal{A}$ then $\mathcal{T}_{G_{\mathcal{A}}}$ is a projector to the algebra $\mathcal{A}$.

%

%

Finally, it is worth noting that if $G$ is a gauge group then the two groups $G$ and $G'$ are dual reductive pairs. However, the inverse is not true, i.e. if two groups are dual reductive pairs, they are not necessarily each other's centralizers in the group of all unitaries. For example, according to the Schur-Weyl duality, the canonical representation of the permutation group on $(\mathbb{C}^d)^{\otimes n}$, i.e. $\textbf{P}( \mathcal{S}_{n})$, and the collective action of $\text{U}(d)$, i.e.  $\textbf{Q}(\text{U}(d))$, are dual reductive pairs but they are surely not equal to one another's centralizer in the group of all unitares acting on $(\mathbb{C}^d)^{\otimes n}$.



\subsection{From gauge groups to dual reductive pairs on product spaces}\label{sec-gauge-dual}

For a subgroup $H$ of $\text{U}(d)$ we denote $H^{\times n}$ to be the group $H^{\times n}\equiv\{U_{1}\otimes \cdots \otimes U_{n}:U_{i}\in H\} $. Also, let $\langle H^{\times n},  \textbf{P}(\mathcal{S}_{n})\rangle$ denote the  group acting on $(\mathbb{C}^{d})^{\otimes n}$ which is generated by the two groups  $ H^{\times n}$ and $\textbf{P}(\mathcal{S}_{n})=\{\textbf{P}(s):s\in \mathcal{S}_{n}\}$. Note that every  element of $\langle H^{\times n},  \textbf{P}(\mathcal{S}_{n})\rangle$ can be written in the canonical form of $W\textbf{P}(s)$ for a unique $W\in H^{\times n}$ and a unique $s\in \mathcal{S}_{n}$. This implies a homomorphism from $\langle H^{\times n},  \textbf{P}(\mathcal{S}_{n})\rangle$  to $\textbf{P}(\mathcal{S}_{n})$ with the kernel $ H^{\times n}$, and therefore  $\langle H^{\times n},  \textbf{P}(\mathcal{S}_{n}) \rangle=H^{\times n} \rtimes   \textbf{P}(\mathcal{S}_{n})$.

Then one can prove the following generalization of Schur-Weyl duality

\begin{theorem} \label{cor-pairs} (\textbf{Generalization of Schur-Weyl duality})   Suppose $G$ and $G'$ are one another's centralizers in the group of  unitaries $\text{U(d)}$. Then  the  following two algebras are commutants of one another in $\text{End}((\mathbb{C}^d)^{\otimes n})$
\begin{enumerate}
\item $\text{Alg}\{\textbf{Q}(G)\}$, the complex algebra spanned  by $\textbf{Q}(G)$.
\item $\text{Alg}\left\{(G')^{\times n},\textbf{P}(\mathcal{S}_{n})\right\}$, the complex  algebra spanned by   $\langle (G')^{\times n},  \textbf{P}(\mathcal{S}_{n})\rangle$.
\end{enumerate}
In other words, the subgroups $\textbf{Q}(G)$ and  $\langle (G')^{\times n},  \textbf{P}(\mathcal{S}_{n})\rangle$ are dual reductive pairs in $\text{GL}((\mathbb{C}^{d})^{\otimes n})$.
\end{theorem}
Using Eq.(\ref{comm-tens}) we  can rephrase the theorem as
\begin{corollary}\label{cor-gauge-com}
Let $G_{\mathcal{A}}$ be the gauge group of    the von Neumann algebra      $\mathcal{A}\subseteq \text{End}(\mathbb{C}^{d})$. Then
\begin{equation}
\text{Comm}\{\textbf{Q}(G_{\mathcal{A}})\} = \textrm{Alg}\{ \mathcal{A}^{\otimes n},  \textbf{P}(\mathcal{S}_{n})\}\,.
\end{equation}
\end{corollary}
This form of the theorem is  particularly useful and has a straightforward  physical interpretation  which will be studied in section \ref{sec:intuitive}.

Theorem \ref{cor-pairs} together with the proposition \ref{prop-doub-comm} implies that there is a one-to-one correspondence between the irreps of the group $G$  which show up in representation $\textbf{Q}(G)$ on $(\mathbb{C}^d)^{\otimes n}$ and the irreps of the group $\langle (G')^{\times n},  \textbf{P}(\mathcal{S}_{n})\rangle$ which show up in  this space. Furthermore, the theorem implies that the representation of $\textbf{Q}(G)\times \langle (G')^{\times n},  \textbf{P}(\mathcal{S}_{n})\rangle$ is multiplicity-free on $(\mathbb{C}^d)^{\otimes n}$.
Note that in the specific case of $G=\text{U}(d)$ (where $G'$ is the trivial group) this dual reductive pair reduces to the Schur-Weyl duality (see theorem \ref{Thm-Schur-Weyl duality}).

Also note that the fact that each of the algebras in this theorem is in the commutant of the other algebra is trivial. 
In other words, for any subgroup $H\subseteq \text{U}(d)$ it holds that
$$\text{Alg}\{\textbf{Q}(H)\}\subseteq \text{Comm}\{(H')^{\times n},\textbf{P}(\mathcal{S}_{n})\}$$
The non-trivial content of the theorem is that for gauge groups these two algebras are equal. For $H$ a subgroup of $\text{U}(d)$ that is not equal to the centralizer of its centralizer in $\text{U}(d)$, i.e. $H\neq H''$, and so is not a gauge group, the above two algebras are not necessarily equal.  We provide a simple example illustrating this fact in Appendix \ref{app-count-ex0}. 

To prove theorem \ref{cor-pairs} we use the following property of gauge groups which is proven in section~\ref{Sec-Proof}.
\begin{lemma} \label{Lem-Per}
For a gauge group $G$, the complex algebra spanned by $\textbf{Q}(G)$ is equal to the  permutationally invariant  subalgebra of the complex algebra spanned by $G^{\times n}$.
\end{lemma}
The result can be summarized as
\begin{align*}
G''= G \;\;\Rightarrow\;\; \textrm{Alg}\{\textbf{Q}(G)\} = \textrm{Alg}\{G^{\times n}\}\ \cap\ \text{Comm}\{ \textbf{P}(\mathcal{S}_n)\}\ \\ = \textrm{Alg}\{G\}^{\otimes n}\ \cap\ \text{Comm}\{ \textbf{P}(\mathcal{S}_n)\}.
\end{align*}
Using this lemma the proof of theorem \ref{cor-pairs} is then straightforward and proceeds as follows.

\begin{proof}(\textbf{Theorem} \ref{cor-pairs}) Since both algebras are von Neumann algebras, we only need to show that one is the commutant of the other, the other direction follows from Eq.(\ref{com-com2}).  So to prove the theorem it is sufficient to show that
$\text{Comm}\{ G'^{\times n},\textbf{P}(\mathcal{S}_{n})\}=\text{Alg}(\textbf{Q}(G))$.  To show this, we note that
$$\text{Comm}\{ G'^{\times n},\textbf{P}(\mathcal{S}_{n})\}=\text{Comm}\{ G'^{\times n}\}\ \cap\ \text{Comm}\{\textbf{P}(\mathcal{S}_{n})\}.$$
Then  since $\text{Comm}\{ G'^{\times n}\}=\text{Alg}\{G^{\times n}\}$
we conclude that
$$\text{Comm}\{ G'^{\times n},\textbf{P}(\mathcal{S}_{n})\}=\text{Alg}\{G^{\times n}\}\ \cap \text{Comm}\{\textbf{P}(\mathcal{S}_{n})\}.$$

This together with lemma \ref{Lem-Per} completes the proof of theorem.
\end{proof}

Finally, it is worth mentioning the following corollary of lemma \ref{Lem-Per} which applies to arbitrary subgroup of $\text{U}(d)$
\begin{corollary} \label{lem-twr}
For any unitary subgroup $H\subseteq \text{U}(d)$,  the permutationally invariant subalgebra of $\text{Comm}\{H^{\times n}\}$ is equal to $\text{Alg}\{\textbf{Q}(H')\}$.
%
\end{corollary}

\begin{proof}
First note that Eq.(\ref{comm-alg})  together with  the commutation theorem for tensor products, i.e. Eq.(\ref{com-tensor-product}), implies
$$\text{Comm}\{H^{\times n}\}=\text{Alg}\{(H')^{\times n}\}$$
Then, from section \ref{sec:charac} we know that the centralizer of  $H$ an arbitrary subgroup of $\text{U}(d)$ is a gauge group and so one can apply lemma \ref{Lem-Per} for gauge group $H'$ which implies that the permutationally invariant subalgebra of $\text{Alg}\{(H')^{\times n}\}$ is equal to $\text{Alg}\{\textbf{Q}(H')\}$.
 \end{proof}

\subsection{An intuitive account}\label{sec:intuitive}

Our generalization of Schur-Weyl duality appears very intuitive if one considers a particular problem concerning two independent observers using different conventions to describe quantum systems.

Suppose that Alice and Bob each use their own personal convention to associate observables with operators in the Hilbert space of a system, and assume that each observer  is not aware of the other's convention.  All they know is that for a particular set of operators $\{A_{i}\}$, the observable which is described by operator $A_{i}$ relative to Alice's convention is also described by operator $A_{i}$ relative to Bob's convention.  Clearly, Alice and Bob will also agree on any observable which is an algebraic function of the $\{A_{i}\}$ and the identity operator $I$, so the full set of observables on which they agree are those in the algebra $\mathcal{A} \equiv \text{Alg}\{A_{i}, I\}$. \footnote{Here we assume the two observers have agreement on the notion of time direction such that the relation between their reference frames can be described by a unitary rather than an anti-unitary.} The question is: what sorts of states and observables can they agree upon for the composite of $n$ systems, assuming Alice and Bob use the same convention for each system and agree on how to label the systems?

It is obvious that Alice and Bob agree on the description of all observables which are in  the algebra generated by (i) the $n$-fold tensor product of the algebra $\mathcal{A}$ and (ii) the canonical representation of the permutation group.  Furthermore, it is intuitively clear that there are no other observables in addition to these that they can agree upon.

Now note that the group $G_{\mathcal{A}}$ of unitaries that commute with $\mathcal{A}$ can be interpreted as the possible ways in which Alice and Bob's conventions may be related to one another.  Because Alice and Bob use the \emph{same} convention for each of the $n$ systems, the operators that they can agree on for the composite are those that are invariant under the \emph{collective} action of $G_{\mathcal{A}}$, i.e. $\mathbf{Q}(G_{\mathcal{A}})$. What is intuitively clear, therefore, is that the operators that are in the commutant of the collective action of $G_{\mathcal{A}}$ are those in the algebra spanned by the $n$-fold product of $\mathcal{A}$, $\mathcal{A}^{\otimes n}$, and the canonical representation of the permutation group, $\mathbf{P}(\mathcal{S}_n)$.  But this is precisely the content of our generalization of Schur-Weyl duality, in the form presented in corollary \ref{cor-gauge-com}.\footnote{Indeed, it was in attempting to make this intuition rigorous that we were led to prove the duality.} We discuss more on this physical interpretation of our generalization of Schur-Weyl duality in \cite{Agr-Marv-Spek-2011}.

This discussion also reveals the motivation for calling the group $G_{\mathcal{A}}$ a \emph{gauge group}.  It is because such a group describes the possible transformations that leave the physically relevant set of observables invariant (in this case, the single-system observables that Alice and Bob agree upon), and such transformations are typically called gauge transformations by physicists.

\subsection{Duality within the symmetric and antisymmetric subspaces}\label{sec:restric}

In the special case where the support of operators are restricted to the symmetric or anti-symmetric subspace, theorem \ref{cor-pairs} has an interesting corollary.  Let $\Pi_{\pm}$ be the projector to $\left[(\mathbb{C}^{d})^{\otimes n}\right]_{\pm}$, the symmetric (respectively antisymmetric) subspace of $(\mathbb{C}^{d})^{\otimes n}$. Then we can prove that

\begin{theorem} \label{Thm-sym}
Suppose $G$ and $G'$ are one another's centralizers in the group of  unitaries $\text{U(d)}$. Then the following two algebras are the commutants of one another in $\text{End}( \left[(\mathbb{C}^d)^{\otimes n}\right]_{\pm})$
\begin{enumerate}
\item $\text{Alg}\{\Pi_{\pm}\textbf{Q}(G)\Pi_{\pm}\}$, the complex  algebra spanned  by $\Pi_{\pm}\textbf{Q}(G)\Pi_{\pm}$.
\item $\text{Alg}\{\Pi_{\pm}\textbf{Q}(G')\Pi_{\pm}\}$, the complex  algebra spanned by   $\Pi_{\pm}\textbf{Q}(G')\Pi_{\pm} $.
\end{enumerate}
In other words, $\Pi_{\pm}\textbf{Q}(G)\Pi_{\pm}$ and  $\Pi_{\pm}\textbf{Q}(G')\Pi_{\pm}$ are dual reductive pairs in $\text{GL}(\left[(\mathbb{C}^d)^{\otimes n}\right]_{\pm})$.
\end{theorem}

Again, the fact that each of these algebras is in the commutant of the other is trivial. The non-trivial fact is that each is \emph{equal} to the commutant of the other.
 We can summarize the theorem by
\begin{equation}
G''=G \;\;\implies\;\; \text{Comm}\{ \Pi_{\pm} \textbf{Q}(G)\Pi_{\pm}\}  = \textrm{Alg}\{\Pi_{\pm} \textbf{Q}(G') \Pi_{\pm}\} \,.
\end{equation}
where here by $\text{Comm}\{ \Pi_{\pm} \textbf{Q}(G)\Pi_{\pm}\}$ we mean the set of all  operators in $\text{End}( \left[(\mathbb{C}^d)^{\otimes n}\right]_{\pm})$ which commute with $\Pi_{\pm} \textbf{Q}(G)\Pi_{\pm}$.

\begin{proof}(\textbf{Theorem} \ref{Thm-sym})
Again since both algebras are von Neumann algebra, we only need to show 
that $\text{Comm}\{\Pi_{\pm} \textbf{Q}(G')\Pi_{\pm}\}=\text{Alg}\{\Pi_{\pm}\textbf{Q}(G)\Pi_{\pm}\}$.
Let $M$ be an arbitrary operator  in $\text{End}( (\mathbb{C}^d)^{\otimes n})$ such that $\Pi_{\pm}M\Pi_{\pm}$ commutes with  $\Pi_{\pm}\textbf{Q}(G')\Pi_{\pm}$. 
Then $\Pi_{\pm} M\Pi_{\pm}$ clearly commutes with $\textbf{Q}(G')$ and therefore theorem \ref{cor-pairs} implies that it is in the span of $\langle G^{\times n},  \textbf{P}(\mathcal{S}_{n})\rangle$. Now recall that,  every arbitrary element of $\langle G^{\times n},  \textbf{P}(\mathcal{S}_{n})\rangle$ can be written in the canonical form of $W\textbf{P}(s)$ for a unique $W\in G^{\times n}$ and a unique $s\in \mathcal{S}_{n}$. So
\begin{equation}
\Pi_{\pm}M\Pi_{\pm}=\sum_{W\in G^{\times n},\  s\in\mathcal{S}_{n} } c_{W, s}\ \  W\textbf{P}(s)
\end{equation}
for some complex coefficients  $c_{W, s}$. Then 
\begin{equation}
\Pi_{\pm}M\Pi_{\pm}=\Pi_{\pm}\left[ \sum_{W\in G^{\times n},\  s\in\mathcal{S}_{n} } (-1)^{p_{\pm}(s)} \ c_{W, s}\ \  W\right]\Pi_{\pm}
\end{equation}
where $\textbf{P}(s)\Pi_{\pm}=(-1)^{p_{\pm}(s)} \Pi_{\pm}$ for arbitrary $s\in\mathcal{S}_{n}$, $(-1)^{p_{+}(s)}=1$ for all $s\in \mathcal{S}_{n}$ and $(-1)^{p_{-}(s)}=\pm1$ dependent on whether $s$ is an odd or even permutation. Therefore, there exists an operator $\bar{M}$ in the span of $G^{\times n}$ such that $\Pi_{\pm}\bar{M}\Pi_{\pm}=\Pi_{\pm}M\Pi_{\pm}$. Then
\begin{equation}
\Pi_{\pm} \left[\sum_{s\in\mathcal{S}_{n} }  \textbf{P}(s) \bar{M}  \textbf{P}^{\dag}(s) \right]\Pi_{\pm}=\Pi_{\pm} \bar{M}\Pi_{\pm}=\Pi_{\pm}M\Pi_{\pm}
\end{equation}
where we have used the fact that $\Pi_{\pm}\textbf{P}(s)=\textbf{P}^{\dag}(s)\Pi_{\pm}=(-1)^{p_{\pm}(s)}\Pi_{\pm}$ and the two negative signs cancel each other. Since $\bar{M}$ is in the span of  $G^{\times n}$ then $\tilde{M}\equiv\sum_{s\in\mathcal{S}_{n} }  \textbf{P}(s) \bar{M}  \textbf{P}^{\dag}(s)$ is in the permutationally invariant subalgebra of the span  $G^{\times n}$. Now since $G$ is a gauge group, using lemma \ref{Lem-Per} we can conclude that $\tilde{M}\in\text{Alg}\{\textbf{Q}(G)\}$. So for any arbitrary $M\in\text{End}( (\mathbb{C}^d)^{\otimes n})$  if $\Pi_{\pm} M\Pi_{\pm}$ commutes with $\Pi_{\pm}\textbf{Q}(G')\Pi_{\pm}$ then there exists an operator $\tilde{M}$  in $\text{Alg}\{\textbf{Q}(G)\}$ such that $\Pi_{\pm}\tilde{M}\Pi_{\pm}=\Pi_{\pm}{M}\Pi_{\pm}$. This completes the proof of the theorem.
\end{proof}


Again, using the proposition \ref{prop-doub-comm} one can see that theorem \ref{Thm-sym} implies: (i) a one-to-one correspondence  between the irreps of $G$ which show up in the representation $\textbf{Q}(G)$ in the symmetric (antisymmetric) subspace and the irreps of $G'$ which show up in the representation $\textbf{Q}(G')$ in the symmetric (antisymmetric) subspace, and (ii) that in these subspaces $\textbf{Q}(G)\times \textbf{Q}(G')$ is multiplicity-free. The special case of this result  is known in the representation theory for the case of the symmetric subspace of $(\mathbb{C}^{d_{1}d_{2}})^{\otimes n}$ and  the collective representations of $G=\text{U}(d_{1})\times e$ and $G'=e \times \text{U}(d_{2})$  as two subgroups of $\text{U}(d_{1}d_{2})$

Applying theorem \ref{Thm-sym} for $G_{\mathcal{A}}$ the gauge group of a von Neumann algebra $\mathcal{A}$  one can show that for any given operator $\Pi_{\pm} M\Pi_{\pm}$ which commutes with $\textbf{Q}(G_{\mathcal{A}})$,  there is an operator $\tilde{M}_{\pm}$ in the permutationally invariant subalgebra of $\mathcal{A}^{\otimes n}$ such that
$$\Pi_{\pm}\tilde{M}_{\pm}\Pi_{\pm}=\Pi_{\pm}{M}\Pi_{\pm}.$$
However, this argument is not constructive and for a given $M$ it is not clear how we can find such an  operator $\tilde{M}_{\pm}$ with this property. In the following theorem, we introduce a completely positive unital quantum operation which does this transformation.




\begin{theorem} \label{pro-map-s}
Let $G_{\mathcal{A}}\subseteq \text{U}(d)$ be the gauge group of a von Neumann algebra $\mathcal{A}\subseteq \text{End}(\mathbb{C}^{d})$. Then there exists a superoperator $\mathcal{L}_{\pm}$ from $\text{End}\left((\mathbb{C}^d)^{\otimes n}\right)$ to itself such that
 \begin{enumerate}
 \item  $\mathcal{L}_{\pm}$ is unital and completely positive,
 \item  The image of $\mathcal{L}_{\pm}$ is in the permutationally invariant subalgebra of $\mathcal{A}^{\otimes n}$
and
\item if $\Pi_{\pm}M \Pi_{\pm}$ commutes with $\textbf{Q}(G_{\mathcal{A}})$
then
$$\Pi_{\pm}\mathcal{L}_{\pm}{(M)} \Pi_{\pm}=\Pi_{\pm}M \Pi_{\pm} \label{eq:actionofL}$$
\end{enumerate}
An instance of such a superoperator is given by
\begin{equation}\label{M-tilde1}
\mathcal{L}_{\pm}{(\cdot)}\equiv \Phi_{\pm}(\cdot) + \mathrm{tr}(\cdot) \frac{\mathbb{I}^{\otimes n} - \Phi_{\pm}(\mathbb{I}^{\otimes n}) }{d^n} 
\end{equation}
with
\begin{equation}\label{M-tilde1prime}
\Phi_{\pm}{(\cdot)}\equiv  \bigoplus_{\mu}  p^{-1}_{\mu,\pm} \ P_{\mu}[\mathcal{T}_{G_{\mathcal{A}}}^{\otimes n}  (\Pi_{\pm} (\cdot)\Pi_{\pm} )]P_{\mu}
\end{equation}
where $\mu$ labels all the irreps of $G'_{\mathcal{A}}$ which show up in the representation $\textbf{Q}(G'_{\mathcal{A}})$, $P_{\mu}$ is the projector to the subspace of $(\mathbb{C}^d)^{\otimes n}$ associated to irrep $\mu$,  $p_{\mu,\pm} \equiv {\text{tr}}\left(P_{\mu} \mathcal{T}_{G_{\mathcal{A}}}^{\otimes n}(\Pi_{\pm}) \right)$
 and the summation in Eq.~(\ref{M-tilde1prime}) is over all the irreps $\mu$ for which $p_{\mu}$ is nonzero.
\end{theorem}
This is proven in Appendix ~\ref{Sec-Proof}. This theorem will be particularly useful in the rest of this paper.

\section{General applications in Quantum Information}\label{Gen-appl}

Schur-Weyl duality has many applications in quantum information theory and so we expect that this generalization will as well.  Here we present two specific important examples of these applications. The first example is about finding noiseless subystems for collective noise associated with a gauge group, and the second is about how, for $n$ copies of a system in a pure state confined to the symmetric or antisymmetric subspace, a measurement with \emph{global symmetry} relative to a gauge group can be simulated by one that has \emph{local symmetry} for that group.  This second result is the seed of the next section, where we will consider the consequences for multi-copy estimation problems in more depth.


\subsection{Characterizing the multi-partite operators that are globally symmetric}\label{sec-charac-glob}

Many applications of Schur-Weyl duality in quantum information theory are based on the fact that it provides a simple characterization of all operators in $\text{End}\left((\mathbb{C}^{d})^{\otimes n}\right)$  which commute with $\textbf{Q}(\text{U}(d))$ or conversely  all operators in $\text{End}\left((\mathbb{C}^{d})^{\otimes n}\right)$  which commute with $\textbf{P}(\mathcal{S}_{n})$.

Theorem \ref{cor-pairs}  and its corollary \ref{cor-gauge-com} immediately yield a characterization of operators with global symmetry under a gauge group   $G$, i.e. the operators in $\text{End}\left((\mathbb{C}^{d})^{\otimes n}\right)$  which commutes with $\textbf{Q}(G)$ -- they lie in the span of the local action of $G'$, i.e. $G'^{\times n}$, and the action of the permutation group, i.e. $\textbf{P}(\mathcal{S}_{n})$.
Similarly, theorem \ref{Thm-sym} 
yields a characterization of operators confined to the symmetric and antisymmetric subspaces that have global symmetry under $G$.  These are simply the operators in the span of the collective action of $G'$.

A straightforward application of this characterization is to find noiseless subsystems. In the following we present a simple example of this.





\subsubsection{Example: Finding noiseless subsystems}\label{sec:noiseless}

We begin by reviewing the standard story about noiseless subsystems. Suppose one is going to send quantum information through a noisy qubit channel, where the noise is described by a unitary that is sampled at random, but wherein the same unitary acts on each qubit.  For example, the qubits could be spin-half particles with a nonzero magnetic moment and the noise could be due to a random magnetic field. As another example, the qubits could be realized as the polarization of photons sent through a fiber-optic cable and the noise could be due to random strains in the cable that induce changes in the polarization. In many cases, it is a good approximation to assume that the noise varies slowly compared to the interval between the qubits as they pass down the channel (or that it varies little on the distance scale between the qubits in the case of a quantum memory), in which case one can assume that the same random unitary is applied to all $n$ qubits. Then it turns out that, due to the symmetry of the noise, it is possible to encode classical and quantum information in the $n$ qubit system in such a way that it remains unaffected by the noise \cite{Noiseless1,Noiseless1',Noiseless2, Noiseless2',BRS03}. To see this, note that under these assumptions, the noise is described by the group $\textbf{Q}(\text{U}(2))$. Any state in the commutant of  $\textbf{Q}(\text{U}(2))$ is invariant under the noise. Furthermore, any state in the span of  $\textbf{P}(\mathcal{S}_{n})$ has this property as well.  Now using Schur-Weyl duality one can conclude that the span of $\textbf{P}(\mathcal{S}_{n})$ is equal to the commutant of $\textbf{Q}(\text{U}(2))$ and therefore \emph{every} state which is unaffected by this type of noise is in the span of $\textbf{P}(\mathcal{S}_{n})$.


In a more general model,  the system sent through the channel may have other degrees of freedom which can potentially be used to send quantum information. In other words, the Hilbert space describing each particle sent through the channel is not $\mathbb{C}^{2}$ but it is $\mathbb{C}^{2} \otimes \mathcal{H}$ where the finite dimensional Hilbert space $\mathcal{H}$ describes another degree of freedom which is invariant under the noise in the channel.  For example, in the case of photons  one can use time-bin encoding in addition to the polarization encoding to encode an extra qubit in each photon. But the time-bin qubit does not suffer from depolarization or polarization mode-dispersion. In other words, this degree of freedom is invariant under polarization noise.

So we assume the noise in the channel is described by a random unitary in the form of $V\otimes \mathbb{I}_{\mathcal{H}}$ where $V\in\text{U}(2)$ and it acts on $\mathbb{C}^{2}$ and $\mathbb{I}_{\mathcal{H}}$ is the identity operator acting on $\mathcal{H}$. In the case of a single system sent through the channel  ($n=1$), it is clear that any information encoded in the subsystem $\mathcal{H}$ is preserved under this type of noise.
Consider the case of many systems sent through the channel ($n>1$).  The question is what are the set of all states of the $n$ systems which are invariant under this type of noise. In other words, what is the set of all states which commute with $\textbf{Q}(\text{U}(2)\otimes  \mathbb{I}_{\mathcal{H}})$? Clearly, in this case, the usual form of Schr-Weyl duality does not apply. But one can use the generalization of Schur-duality we presented in the previous section to find these density operators.

To see this, first note that the group of unitaries   $G\equiv\{V\otimes \mathbb{I}_{\mathcal{H}}:V\in\text{U}(2)\}$ is the gauge group of the algebra $ \mathbb{I}_{2}\otimes \text{End}(\mathcal{H})$ where $\mathbb{I}_{2}$ is the identity operator on $\mathbb{C}^{2}$. Then corollary \ref{cor-gauge-com} (which is indeed another version of theorem \ref{cor-pairs}) gives the characterization of all operators which commute with $\textbf{Q}(G)$: These are the set of all operators in $\text{Alg}\{\textbf{P}(\mathcal{S}_{n}), \left(\mathbb{I}_{2}\otimes \text{End}(\mathcal{H})\right)^{\otimes n}\}$. So the set of all density operators in this algebra is exactly the set of all states which remain unaffected under this type of noise.  This means that to protect information one needs to encode it in either the invariant degree of freedom of each subsystem ($\mathcal{H}$) or in the permutational degree of freedom.  Again note that even without using our results,  it is straightforward to see that all of these states remain unchanged under this noise. The non-trivial consequence of the corollary is that this algebra includes all such density operators.

Note that if the group $H\subseteq \text{U}(d)$ describing the noise is not a gauge group then the $\text{Comm}\{\textbf{Q}(H)\}$ can be larger than $\text{Alg}\{(H')^{\otimes n},\textbf{P}(\mathcal{S}_{n})\}$ as it is shown by a simple example in section \ref{app-count-ex0} (where the group $H$ is the $j=1$ representation of $\text{SU}(2)$ in $\mathbb{C}^{3}$). This means that, unlike the case of noise described by a gauge group,  one can encode quantum information in a space which is larger than the permutational degree of freedom of the systems together with the invariant degrees of freedom of each system.

\subsection{Promoting global symmetries to local symmetries}\label{sec:Promote}

Another important application of this new duality is that in particular cases one can \emph{promote  a global symmetry to local symmetry} as we will describe in this section.

Recall the definition of local and global symmetries for an arbitrary operator $M \in \text{End}{(\mathbb{C}^d)^{\otimes n}}$.  $M$ has a global symmetry with respect to the symmetry group $H\subseteq \text{U}(d)$ if $M$ is invariant under the collective action of $H$, as specified in Eq.~\eqref{global}, i.e. if $M\in \text{Comm}\{\textbf{Q}(H)\}$.  Similarly we say that $M$ has local symmetry with respect to the symmetry group $H$ if it is invariant under the local action of $H$, as specified in Eq.~\eqref{local}, i.e. if $M\in \text{Comm}(H^{\times n})$.

As noted in the introduction, 
the condition of local symmetry is generally much stronger than global symmetry.  For example, if $H$ is the group of rotations then  global symmetry of a Hamiltonian with respect to $H$ implies that the vector of the total angular momentum of $n$ systems is a constant of the motion in the dynamics generated by this Hamiltonian.  However, in this case the angular momenta of the subsystems are not necessarily conserved and  the $n$ subsystems can exchange angular momentum with one another. On the other hand, having a Hamiltonian with local symmetry with respect to the group of rotation implies the existence of non-trivial constants of motion defined on each of the $n$ subsystems. So in this case we will have $n$ conserved vectors of angular momentum and under this type of Hamiltonian, subsystems cannot exchange angular momentum.

Now consider the case where the symmetry under consideration is described by a gauge group $G_{\mathcal{A}}$ of a von Neumann algebra $\mathcal{A}\subseteq \text{End}(\mathbb{C}^{d})$.
Note that if $M\in \text{End}((\mathbb{C}^d)^{\otimes n})$ has global symmetry with respect to $G_{\mathcal{A}}$ then $\Pi_{\pm}M\Pi_{\pm}$ will also have global symmetry with respect to $G_{\mathcal{A}}$. Then according to theorem \ref{pro-map-s} for any operator $M$ with global symmetry with respect to $G_{\mathcal{A}}$ there is an operator $\tilde{M}_{\pm}$ which has local symmetry with respect to $G_{\mathcal{A}}$ and is equal to $M$ within the symmetric (anti-symmetric) subspace,
$$\Pi_{\pm} \tilde{M}_{\pm} \Pi_{\pm}=\Pi_{\pm} M \Pi_{\pm} $$
 One can choose $\tilde{M}_{\pm}=\mathcal{L}_{\pm}(M)$  where $\mathcal{L}_{\pm}$ is the completely positive, unital superoperator defined in theorem \ref{pro-map-s}.  So using the terminology of local and global symmetry we can interpret theorem \ref{pro-map-s} as \emph{promoting global symmetry to local symmetry}.


In the following we explore the important consequence of
promoting global symmetry to local symmetry  for the case of measurements.

\subsubsection{Measurements with Global and Local symmetry}

The most general type of measurements that can be performed on a quantum system  can be described by a POVM (positive operator-valued measure) (See e.g. \cite{Holevo, Chiribella}).  Consider a POVM $M:\sigma(\Omega)\rightarrow \text{End}\left((\mathbb{C}^d)^{\otimes n}\right)$. Here, $\Omega$ denotes the space of outcomes of the measurement. This is a measure space equipped with a $\sigma$-algebra of subsets, denoted by $\sigma(\Omega)$. The elements of the $\sigma$-algebra are subsets of $\Omega$, where  $B\subseteq \Omega$ corresponds to the event that the outcome of measurement is an element of $B$.

We say a POVM $M:\sigma(\Omega)\rightarrow \text{End}\left((\mathbb{C}^d)^{\otimes n}\right)$ has global/local symmetry with respect to the  group $H\subseteq \text{U}(d)$  if for any $B\in \sigma(\Omega)$, the operator $M(B)$ has  global/local symmetry with respect to $H$, i.e. it satisfies Eq.(\ref{global}) or  Eq.(\ref{local}) respectively. Again, typically the local symmetry condition on a measurement is a much more restrictive condition.

In the following we first explore the consequences of a measurement having local symmetry and then we see how in the case of gauge symmetries  using the generalization of Schur-Weyl duality and in particular theorem \ref{pro-map-s}, one can  promote global symmetry of a measurement to a local symmetry (for states whose support is restricted to the symmetric or anti-symmetric subspace). Since the locally symmetric measurements typically are a much smaller class of measurements, this technique will be particularly useful  in quantum estimation problems where one seeks to find the measurement that optimizes some figure of merit.  Also, this trick is useful for determining whether a given estimation problem requires a nonlocal measurement on the $n$ subsystems (i.e. one that requires a quantum channel or entanglement) or whether a local measurement suffices. More generally, it can set an upper bound on the amount of entanglement required to achieve a particular degree of success in estimation.  In the following we explain more about this.

One way to understand the restriction of local symmetry of measurements is via the following observation: Let the subgroup $H$ of $\text{U}(d)$ be a subgroup with unique Haar measure $d\mu$ and consider the twirling  superoperator defined in Eq.(\ref{Def:twirling}). Then  local symmetry  of POVM $M:\sigma(\Omega)\rightarrow \text{End}\left((\mathbb{C}^d)^{\otimes n}\right)$ with respect to $H$ implies that $\mathcal{T}^{\otimes n}_{H}(M)=M$. This in turn implies that for any event $B\in \sigma(\Omega)$ and for any arbitrary density operator $\rho\in\text{End}\left((\mathbb{C}^{d})^{\otimes n}\right)$ it holds that
\begin{align*}
\text{Pr}(B)&=\text{tr}\left(M(B)\rho \right)\\ &=\text{tr}\left(\mathcal{T}^{\otimes n}_{H}\left(M (B)\right)\rho \right)=\text{tr}\left(M(B)\mathcal{T}^{\otimes n}_{H}(\rho) \right)
\end{align*}
In other words, for any arbitrary state $\rho$ if before measurement $M$, we apply the local twirling operation  $\mathcal{T}_{H}$, then we do not disturb the statistics of the measurement $M$. Note that by applying the  twirling operation before the measurement, we are mapping the state to  $\text{Alg}\{H'\}^{\otimes n}$ which typically can be much smaller than the space of all density operators in $\text{End}\left((\mathbb{C}^{d})^{\otimes n}\right)$. Applying this twirling operation decreases the size of the subsystems of the Hilbert space on which the state could be non-trivial and, as we will see later, this fact can set an upper bound on the amount of entanglement required to achieve a particular inference.


This is more clear in the case of gauge  groups. Let $G_{\mathcal{A}}$ be the gauge group of a von Neumann algebra $\mathcal{A}\subseteq \text{End}(\mathbb{C}^{d})$. Then  for any state $\rho$, the state $\mathcal{T}^{\otimes n}_{G_{\mathcal{A}}}(\rho)$ is in  $\mathcal{A}^{\otimes n}$. Using the decomposition of the matrix algebra $\mathcal{A}$ given by Eq.(\ref{decom-algeb}), one can find a simple characterization of the form of state $\mathcal{T}^{\otimes n}_{G_{\mathcal{A}}}(\rho)$ for arbitrary $\rho$.

 For instance, consider the Hilbert space $\mathcal{H}=\mathcal{H}_{L}\otimes \mathcal{H}_{R}$ where  $\mathcal{H}_{L}$ and $\mathcal{H}_{R}$ are two finite-dimensional Hilbert spaces.
 The system of interest decomposes into two subsystems: the left subsystem, described by $\mathcal{H}_{L}$, and the right subsystem, described by $\mathcal{H}_{R}$. Let the von Neumann Algebra $\mathcal{A}$ be $\text{End}(\mathcal{H}_{L})\otimes \mathbb{I}_{\mathcal{H}_R}$ where $\mathbb{I}_{\mathcal{H}_R}\in \text{End}(\mathcal{H}_{R})$ is the identity operator on $\mathcal{H}_{R}$.  As we have seen in the above, for any measurement with local symmetry with respect to the group $G_{\mathcal{A}}$ the statistics of outcomes of the measurement on state $\rho$ is exactly the same as the statistics of the outcomes of that measurement on state $\mathcal{T}^{\otimes n}_{G_{\mathcal{A}}}(\rho)$.  But for any state $\rho\in\text{End}\left(\mathcal{H}^{\otimes n}\right)$, it holds that $\mathcal{T}^{\otimes n}_{G_{\mathcal{A}}}(\rho)\in\mathcal{A}^{\otimes n}=\text{End}(\mathcal{H}_{L}^{\otimes n})\otimes \mathbb{I}^{\otimes n}_{\mathcal{H}_R}$. 
In other words, this means that if before a measurement $M$ with local symmetry with respect to $G_{\mathcal{A}}$, we  discard all the $n$ right subsystems, we still can simulate the measurement $M$ by performing a measurement on the left subsystems.
So, effectively the Hilbert space which is relevant in this problem is $\mathcal{H}^{\otimes n}_{L}$ which is of a smaller size than the Hilbert space $\mathcal{H}^{\otimes n}$. This clearly puts an upper bound on the amount of entanglement required to implement measurement $M$.  We can extend this argument to the case of an arbitrary von Neumann algebra $\mathcal{A}$.

A particularly important case is where $\mathcal{A}$ is a commutative algebra. In this case, for any arbitrary state $\rho\in \text{End}\left((\mathbb{C}^{d})^{\otimes n}\right)$, the state $\mathcal{T}^{\otimes n}_{G_{\mathcal{A}}}(\rho)$, as an element of $\mathcal{A}^{\otimes n}$, commutes with all generators of $\mathcal{A}^{\otimes n}$. So, if on each individual qudit we measure an observable (projective von-Neumann measurement) inside the algebra $\mathcal{A}$ we will not change the state $\mathcal{T}^{\otimes n}_{G_{\mathcal{A}}}(\rho)$.  But since $\mathcal{T}^{\otimes n}_{G_{\mathcal{A}}}(\rho)\in \mathcal{A}^{\otimes n}$, we can uniquely specify  $\mathcal{T}^{\otimes n}_{G_{\mathcal{A}}}(\rho)$ by measuring a set of  observables in $\mathcal{A}$ which generates the algebra $\mathcal{A}$ on each individual qudit (note that  generators of $\mathcal{A}$ all commute with each other and so can be measured simultaneously). However,  after these measurements we know the exact description of the state $\mathcal{T}^{\otimes n}_{G_{\mathcal{A}}}(\rho)$   and so we can then simulate any other measurement by a post-processing of the data we have gathered in these measurements. Finally, we notice that measuring generators of $\mathcal{A}$ on each individual qudit for state $\mathcal{T}^{\otimes n}_{G_{\mathcal{A}}}(\rho)$ gives exactly the same statistics as measuring these generators on the original state $\rho$. So we can summarize this discussion as follows.

\begin{proposition} (\textbf{Commutative Algebras})\label{lem:com}
Let $G_{\mathcal{A}}$ be the gauge group of the commutative von Neumann algebra $\mathcal{A}\subseteq \text{End}(\mathbb{C}^{d})$. Then any measurement on $(\mathbb{C}^{d})^{\otimes n}$ which has local symmetry with respect to $G_{\mathcal{A}}$ can be realized by measuring  a set of observables which generate $\mathcal{A}$ on each qudit  individually followed by a classical processing of the outcomes.
\end{proposition}
Therefore to implement a measurement which has local symmetry with respect to the gauge group $G_{\mathcal{A}}$ of a commutative algebra $\mathcal{A}$ one does not need any entanglement or adaptive measurements.

\subsubsection{From Global to Local symmetry}


Having studied the consequences of local symmetry for measurements, we now show how the result of the previous section and in particular theorem \ref{pro-map-s} implies that for states whose support is restricted to the symmetric/anti-symmetric subspace, the global symmetry of a measurement with respect to a gauge group can be promoted to a local symmetry.

\begin{corollary} (\textbf{Symmetry of Measurements})\label{Thm-Meas-LG Equiv}
Let $G_{\mathcal{A}}$ be the gauge group of a von Neumann algebra $\mathcal{A}\subseteq \text{End}(\mathbb{C}^{d})$. Then for any POVM $M:\sigma(\Omega)\rightarrow \text{End}((\mathbb{C}^d)^{\otimes n})$ which has global symmetry with respect to $G_{\mathcal{A}}$  there is a POVM with local symmetry with respect to $G_{\mathcal{A}}$ (i.e. $\tilde{M}:\sigma(\Omega)\rightarrow \mathcal{A}^{\otimes n})$  which has exactly the same statistics for all states whose supports are confined to the symmetric (anti-symmetric)  subspace. In particular, one can choose  $\tilde{M}_{\pm}=\mathcal{L}_{\pm}(M)$  where $\mathcal{L}_{\pm}$ is the superoperator defined in theorem \ref{pro-map-s}.
\end{corollary}
\begin{proof}
 First, recall that if $N:\sigma(\Omega)\rightarrow\text{End}\left((\mathbb{C}^d)^{\otimes n}\right)$ is a POVM and $\mathcal{E}$ is a unital,  positive quantum operation from $\text{End}((\mathbb{C}^d)^{\otimes n})$ to itself, then $\mathcal{E}(N): \sigma(\Omega)\rightarrow \text{End}((\mathbb{C}^d)^{\otimes n})$ is also a POVM. By theorem \ref{pro-map-s} we know that $\mathcal{L}_{\pm}$ is a unital, completely positive map from $\text{End}\left((\mathbb{C}^d)^{\otimes n}\right)$ to itself. So $\tilde{M}_{\pm}\equiv \mathcal{L}_{\pm}(M)$ where $\tilde{M}_{\pm}:\sigma(\Omega)\rightarrow\text{End}\left((\mathbb{C}^d)^{\otimes n}\right)$ is also a POVM. Furthermore, theorem \ref{pro-map-s} implies that the image of $\mathcal{L}_{\pm}$ has local symmetry with  respect to $G_{\mathcal{A}}$ (i.e. it is in $\mathcal{A}^{\otimes n}$). Finally, by definition, if POVM $M$ has global symmetry with respect to ${G}_{\mathcal{A}}$ then for any $B\in \sigma(\Omega)$, $M(B)$ commutes with $\textbf{Q}(G_{\mathcal{A}})$. Now since all elements of $\textbf{Q}(G_{\mathcal{A}})$ are permutationally invariant they are block diagonal in irreps of the permutation group and in particular they commute with $\Pi_{\pm}$. So if $M(B)$ commutes with $\textbf{Q}(G_{\mathcal{A}})$, then $\Pi_{\pm}M(B)\Pi_{\pm}$ will also commute with $\textbf{Q}(G_{\mathcal{A}})$. Then using  theorem   \ref{pro-map-s} and the definition of $\tilde{M}_{\pm}$ we conclude that for arbitrary event $B\in\sigma(\Omega)$
 \begin{equation}\label{proof-loc-glo}
 \Pi_{\pm} \tilde{M}_{\pm}(B) \Pi_{\pm}=\Pi_{\pm}M(B) \Pi_{\pm}
 \end{equation}
Now consider the probability of event $B\in \sigma(\Omega)$ in the measurement described by POVM $\tilde{M}$ and state $\rho\in \text{End}\left((\mathbb{C}^d)^{\otimes n}\right)$. This probability is  given by $\text{Pr}(B)=\text{tr}( \rho \tilde{M}(B))$. Now if the support of $\rho$ is restricted to the symmetric/anti-symmetric subspace then $\rho=\Pi_{\pm}\rho\Pi_{\pm}$ and so
\begin{align*}
\forall \mu:\ \  \text{Pr}(B)&=  \text{tr}( \rho \tilde{M}(B))= \text{tr}\left(\rho \Pi_{\pm}  \tilde{M}(B) \Pi_{\pm}\right)
\end{align*}
Substituting Eq.(\ref{proof-loc-glo}) into this we conclude that
\begin{align*}
\text{Pr}(B)&= \text{tr}\left(\rho \Pi_{\pm}  \tilde{M}(B) \Pi_{\pm}\right)\\ &=\text{tr}\left(\rho \Pi_{\pm}  {M}(B) \Pi_{\pm}\right)=\text{tr}( \rho M(B))
\end{align*}
But $\text{tr}( \rho M(B))$ is the probability of event $B$ in the measurement described by POVM $M$ performed on state $\rho$. Therefore measurement $\tilde{M}$ simulates measurement $M$.
\end{proof}


Corollary \ref{Thm-Meas-LG Equiv} implies that if the support of state $\rho$  is restricted to the symmetric/anti-symmetric subspace then any measurement with global symmetry with respect to $G_{\mathcal{A}}$ on $\rho$ can be simulated by a measurement on $\mathcal{T}^{\otimes n}_{G_{\mathcal{A}}}(\rho)$. In other words, if one is under the restriction of using measurements which have global symmetry with respect to $G_{\mathcal{A}}$ then by applying  the channel $\mathcal{T}^{\otimes n}_{G_{\mathcal{A}}}$ to a state which is restricted to the symmetric/anti-symmetric subspace one does not lose any information. Note that generally the support of $\mathcal{T}^{\otimes n}_{G_{\mathcal{A}}}(\rho)$  is no longer restricted to the symmetric(anti-symmetric) subspace.

Based on this observation one can put  a strong  condition on the form of measurements which can be  useful, for instance, in finding the optimal measurement in a multi-copy estimation procedure (as we do in the next section). Note that for any given measurement with a global symmetry $G_{\mathcal{A}}$ there are many different other  measurements which will have exactly the same statistics  on all states whose support are restricted to the symmetric/anti-symmetric subspaces. These measurements may require different amounts of entanglement to be implemented. Finding a measurement with local symmetry with respect to $G_{\mathcal{A}}$ in this set of equivalent measurements  has the advantage that one can easily  put an upper bound on the amount of  entanglement required to realize it. In particular, note that the combination of proposition \ref{lem:com} and corollary \ref{Thm-Meas-LG Equiv} implies that if a measurement has global symmetry with respect to $G_{\mathcal{A}}$ the gauge group of a \emph{commutative} algebra $\mathcal{A}$, then among all possible measurements which can simulate this measurements on states with support in symmetric/anti-symmetric  subspace there is one which does not need any entanglement to be realized.

\subsubsection{\text{Example}}

It is useful to consider a concrete example of the simulation of a measurement with global symmetry by one with local symmetry.
To this end, consider a pair of qudits with the total Hilbert space $(\mathbb{C}^{d})^{\otimes 2}$ and consider the unitary group of phase shifts $H_{d}\equiv\{e^{i\phi N}:\phi\in (0,2\pi]\}$ where $N|i\rangle=i |i\rangle$ and $\{|i\rangle: i=0\cdots d-1\}$ is an orthonormal basis for $\mathbb{C}^{d}$. Note that the unitary group $H_{d}$  is indeed a representation of $\text{U}(1)$ on $\mathbb{C}^{d}$.

Now one can easily see that a measurement which has global (local) symmetry with respect to $H_{d}$ has also global (local) symmetry with respect to $ \{e^{i\phi_{0}}e^{i\phi N}:\phi_{0}, \phi\in (0,2\pi]\}$ and vice versa. But in the specific case of $d=2$, the latter group  is a gauge group, as we have seen in section \ref{sec:charac}. In the case of $d=2$ we denote $ \{e^{i\phi_{0}}e^{i\phi N}:\phi_{0}, \phi\in (0,2\pi]\}$  by $G$.



So, in the case of $d=2$ according to  corollary \ref{Thm-Meas-LG Equiv}, we can infer that for states in the symmetric and antisymmetric subspaces, every measurement with global symmetry with respect to $G$ (or equivalently with respect to $H_{2}$) can be simulated with one that has local symmetry with respect to $G$  (or equivalently with respect to $H_{2}$).

The measurements that have local symmetry are those for which all the POVM elements are \emph{locally} diagonal in the eigenspaces of $N$, that is, in the basis $\{|0\rangle,|1\rangle \}$. For a pair of qubits, all such measurements can be realized by a measurement of the basis $\{|00\rangle,|01\rangle,|10\rangle,|11\rangle\}$ followed by a classical post-processing of the outcome.  Note that measurement in basis $\{|00\rangle,|01\rangle,|10\rangle,|11\rangle\}$  can be realized by measuring observable $N$ individually on each qubit. This is expected from  proposition \ref{lem:com} because the algebra of commutants of the gauge group, is the algebra of   diagonal matrices in the basis $\{|0\rangle,|1\rangle\}$ which is a commutative algebra.

On the other hand,  POVM elements  of any measurements that have global symmetry with  respect to $H_{2}$ (or equivalently with respect to $G$) are those which commute with total number operator $N\otimes \mathbb{I} +\mathbb{I} \otimes N$ and so  are block-diagonal relative to the eigenspaces of $N\otimes \mathbb{I} + \mathbb{I} \otimes N$. For example, for any arbitrary $\theta$ the projective measurement in the basis 
$$\{|00\rangle,|11\rangle, \cos\theta |01\rangle+\sin\theta |10\rangle,  \sin\theta |01\rangle-\cos\theta |10\rangle\}$$
has global symmetry with respect to $G$. Note that for all the values of $\theta$ which are not equal to an integer times $\pi/2$ this measurement would be an entangled measurement. 

 Let  $M:\sigma(\Omega)\rightarrow \text{End}\left((\mathbb{C}^{2})^{\otimes 2}\right)$ be the POVM of an arbitrary  measurement on these two qubits which has global symmetry with respect to $G$. Then,  for any arbitrary event $B\in \sigma(\Omega)$, $M(B)$ is block-diagonal relative to the eigenspaces of $N\otimes \mathbb{I} + \mathbb{I} \otimes N$, i.e.
$$\ P_{00}M(B) P_{00}+P_{11}M(B) P_{11}+\left[P_{01}+P_{10}\right]M(B)\left[P_{01}+P_{10}\right]=M(B) $$
where $P_{ij}\equiv|ij\rangle\langle ij|, \ i,j\in\{0,1\}$. Therefore the probability of event $B$ for arbitrary state $\rho$ is equal to
\begin{align*}
\text{tr}\left(M(B) \rho\right)&= \text{tr}\left(P_{00} \rho\right) \text{tr}\left(M(B)P_{00}\right)\\ &+\text{tr}\left(P_{11} \rho\right) \text{tr}(M(B)P_{11})+ \text{tr}\left(\rho \left[P_{01}+P_{10}\right] M(B) \left[P_{01}+P_{10}\right]\right)
\end{align*}
Now if the state $\rho$ is promised to be in the symmetric subspace, i.e. $\Pi_{+}\rho\Pi_{+}=\rho$ then
\begin{align*}
\text{tr}\left(\rho \left[P_{01}+P_{10}\right] M(B) \left[P_{01}+P_{10}\right]\right)&= \text{tr}\left(\Pi_{+}\rho \Pi_{+}  \left[P_{01}+P_{10}\right] M(B) \left[P_{01}+P_{10}\right]\right)\\ &=\text{tr}(\rho  \left[P_{01}+P_{10}\right]) \text{tr}\left(M(B) |\phi^{+}\rangle\langle\phi^{+}|\right)
\end{align*}
where $ |\phi^{+}\rangle\equiv(1/\sqrt{2})(|01\rangle+|10\rangle)$. In other words,
\begin{align*}
\text{tr}(M(B) \rho)=&\text{Pr}(B|00) \text{tr}(P_{00} \rho)+  \text{Pr}(B|11)\text{tr}(P_{11} \rho)+ \text{Pr}(B|01,10)\text{tr}(\rho \left[P_{01}+P_{10}\right])
\end{align*}
where
\begin{align*}
\text{Pr}(B|00) \equiv  \text{tr}\left(M(B)P_{00}\right),\ \ &\text{Pr}(B|11)\equiv\text{tr}\left(M(B)P_{11}\right)\\ \text{and}\ \   &\text{Pr}(B|01,10)\equiv \text{tr}\left(M(B) |\phi^{+}\rangle\langle\phi^{+}|\right)
\end{align*}
and they can be interpreted as the conditional probability of  event $B\in \sigma(\Omega)$ given each of the four outcomes. This means that to simulate this measurement one can measure $N$ individually on each qubit, i.e. project the state of each qubit to the $\{|0\rangle,|1\rangle\}$ basis, and based on the outcomes of these measurements choose an outcome  $\omega\in \Omega$ consistent with these conditional probabilities.




In other words, although the set of  measurements with global symmetry is much larger than the set of measurements with local symmetry,  all the information we can extract using a measurement with global symmetry can also be obtained by a measurement with local symmetry. Note that in this example though implementing the measurement with global symmetry may require entanglement, implementing the measurement with local symmetry does not, nor does it require communication among the subsystems. Also, note that from corollary \ref{Thm-Meas-LG Equiv} we know that this result holds for any arbitrary number of qubits.

It is worth mentioning that  the measurement with local symmetry which we built based on the original measurement is exactly the same measurement as we can get by applying the super-operator $\mathcal{L}_{+}$ defined in theorem \ref{pro-map-s} to the POVM of the original measurement.


 Finally, based on this example we provide another concrete instance that illustrates how the gauge property of the symmetry group is critical for being able to promote global symmetries to local symmetries.  Consider the above example for the case of $d=3$, i.e. for  qu\emph{trits}  rather than qubits. In this case,   $N|i\rangle=i |i\rangle$ where $\{|i\rangle: i=0,\cdots ,2\}$. Then, one can easily see that the group $ \{e^{i\phi_{0}}e^{i\phi N}:\phi_{0}, \phi\in (0,2\pi]\}$ is no longer a gauge group. So, in general  a  measurement  on two qutrits with global symmetry with respect to this group, cannot be necessarily simulated by  a measurement with local symmetry with respect to this group, even under the promise that the state is restricted to the symmetric subspace.


In fact, in this case all the measurements that have local symmetry are those which can be obtained by classical post-processing of a measurement of the product basis $\{ |ij\rangle: i,j=0,1,2 \}$, while those with global symmetry are merely block-diagonal with respect to the eigenspaces of $N\otimes \mathbb{I} + \mathbb{I} \otimes N$.  In particular, a measurement with global symmetry may include the rank-1 projectors onto the vectors 
$|11\rangle + (|02\rangle +|20\rangle)$ and $|11\rangle - (|02\rangle +|20\rangle)$
which both lie in the symmetric subspace.  Such a measurement cannot be simulated by any measurement with local symmetry, which necessarily is unable to detect coherence between $|11\rangle$ and  $|02\rangle +|20\rangle$.

\section{Multi-copy estimation and decision problems}\label{sec:estimation}


The main application of the duality is to multi-copy estimation problems. We begin by setting up a general framework for such problems.

Suppose  Alice randomly chooses a qudit state $\rho$ from the density operators in $\textrm{End}(\mathbb{C}^d)$  according to the probability density function $p$  and then prepares $n$ copies of this state and sends them to Bob through a quantum channel $\mathcal{E}:\text{End}((\mathbb{C}^d)^{\otimes n})\rightarrow \text{End}((\mathbb{C}^d)^{\otimes n})$.  Bob's goal is to estimate some parameter(s) of state $\rho$. (We here adopt the convention that the term ``estimation problem'' includes decision problems as a special case). So upon receiving $n$ systems he performs a measurement and generates some outcome in the outcome space $\Omega$ where  $\Omega$  is a measure space, i.e. a set equipped with a $\sigma$-algebra $\sigma(\Omega)$ of subsets. The elements of the $\sigma$-algebra are subsets of $\Omega$, where  $B\subseteq \Omega$ corresponds to the event that Bob's measurement outcome is an element of $B$. The outcome space $\Omega$  can be continuous (in the case of general estimation problems) or discrete (in the case of decision problems).

\begin{figure} [h]
\begin{center}
\includegraphics[scale=.70]{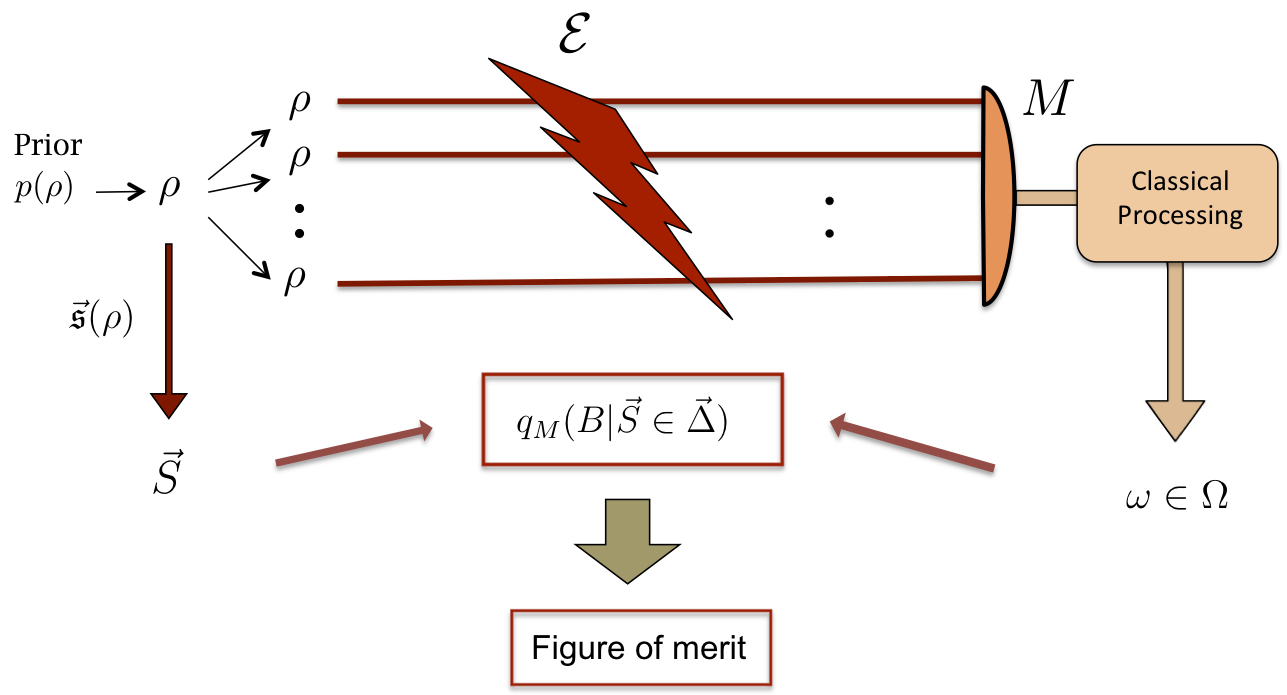}
\caption{Multi-copy estimation problem (see below).}\label{fig:estimation}
\end{center}
\end{figure}

In an arbitrary estimation strategy, Bob measures the $n$ qudits he has received and possibly does some post-processing on the outcome, ultimately generating an output in the set $\Omega$.  The entire strategy, which combines the measurement and the data processing, can be described by a POVM $M:\sigma(\Omega)\rightarrow \text{End}((\mathbb{C}^d)^{\otimes n})$. For simplicity, we will often refer to the estimation strategy as the measurement.


Therefore, the most general figure of merit which evaluates the performance of different strategies in  an estimation problem is a function which assigns real numbers to all POVMs $M:\sigma(\Omega)\rightarrow \text{End}((\mathbb{C}^d)^{\otimes n})$. Equivalently, in the case of the multi-copy  estimation problems  we are considering here, the most general figure of merit can be described as a real functional which acts on the two-variable function
\begin{equation}\label{thesis-figure-merit}
q_M(B|\rho)=\text{tr}\left(M({B})\mathcal{E}(\rho^{\otimes n})\right)
\end{equation}
the conditional probability that, using the strategy described by  POVM $M:\sigma(\Omega)\rightarrow \text{End}((\mathbb{C}^d)^{\otimes n})$, the event $B\in \sigma(\Omega)$ happens given that  Alice has chosen the state $\rho\in  \text{supp}(p)$  and has sent state $\rho^{\otimes n}$ to Bob through the channel $\mathcal{E}$ (here, $\text{supp}(p)$ denotes the support of the distribution $p$).

This describes the most general figure of merit one can define for the multi-copy estimation problems we are considering here. However, in the particular cases where for example the goal is to estimate some parameter of $\rho$, say the expectation value of some observable for state $\rho$, one might use a figure of merit which only depends on the conditional probability of outcomes for different  values of that parameter. Here, we think of the parameter as a random variable defined as a function of the  state Alice chooses each time (The state is random and so any function of the state can be thought of as a random variable).  Let $\mathfrak{s}: \text{supp}(p)\rightarrow \mathbb{R}$ be an arbitrary function from states in $\text{supp}(p)$ to real numbers.   Then this function will map the random state $\rho$ chosen by Alice to a random real variable $S=\mathfrak{s}(\rho)$.
  \ Then if  Bob's goal is to estimate the value of parameter $\mathfrak{s}(\rho)$ for the state $\rho$  which Alice has chosen each time (or to make a decision based on the value of this parameter) a reasonable family of figures of merit to evaluate Bob's performance can be expressed as functionals of
$$q_{M}(B|{S}\in {\Delta}),$$
where ${\Delta}$ is an interval of $\mathbb{R}$.  This is the conditional probability that, using the strategy described by POVM  $M:\sigma(\Omega)\rightarrow\text{End}((\mathbb{C}^d)^{\otimes n})$, event $B$ happens given that the value of the random variable  $S$ is in ${\Delta}$.

On the other hand, one can imagine situations where, for example, the cost for wrong estimation of a parameter ${S}$ not only depends on the  estimated value of ${S}$ and its actual value  but also depends on the value of some other parameter, say ${S}'$, where ${S}'$ is the random variable induced by the function $\mathfrak{s}': \text{supp}(p)\rightarrow \mathbb{R}$ acting on the random state Alice chooses. For instance, one may imagine situations where the cost of wrong estimation of a parameter ${S}$ depends also on the energy of the state, $\text{tr}(\rho H)$, where $H$ is the Hamiltonian. So in this case $\mathfrak{s}'(X)=\text{tr}\left(X H\right)$ defines a relevant parameter to evaluate the performance of the estimation procedure. 

In general, let $$\vec{\mathfrak{s}}(\cdot)=\left(\mathfrak{s}^{(1)}(\cdot),\cdots,\mathfrak{s}^{(l)}(\cdot)\right)$$ be a set of functions where each $\mathfrak{s}^{(i)}(\cdot)$ is a function from $\text{supp}(p)$  to $\mathbb{R}$. Then based on the set of functions $\vec{\mathfrak{s}}(\cdot)=\left(\mathfrak{s}^{(1)}(\cdot),\cdots,\mathfrak{s}^{(l)}(\cdot)\right)$ we can define  a set of random variables $\left({S}^{(1)},\cdots,{S}^{(l)}\right)$ where the random variable $S^{(i)}$ is $\mathfrak{s}^{(i)}(\rho)$ where $\rho$ is the random state Alice has chosen at each round. So a general figure of merit can be expressed as a functional of
$$q_{M}(B|\vec{S} \in \vec{\Delta}),$$
where $\vec{\Delta}$ is an $l$-dimensional interval of $\mathbb{R}^l$.  This is the conditional probability that with Bob's strategy described by POVM  $M:\sigma(\Omega)\rightarrow\text{End}((\mathbb{C}^d)^{\otimes n})$ event $B$ happens  given that the value of the random variables  $\vec{S}$ are in $\vec{\Delta}$.

The other reason to consider $q_{M}(B|\vec{S}\in \vec{\Delta})$ for more than one parameter $S^{(i)}$ is to study the cases where Bob is interested in estimating more than one parameter of the state.

Note that by having a larger number of parameters $l$ we can describe more and more general types of figure of merit.  In general, if $d$ is the dimension of $\mathbb{C}^d$ then the set of all (normalized) density operators can be specified by $d^{2}-1$ real parameters. So having $l=d^{2}-1$ real parameters is sufficient to specify the exact density operator Alice has chosen each time, and so $l=d^{2}-1$ parameters  are sufficient to describe the most general form of figures of merit one can imagine for this problem.
\ However, generally, having a figure of merit which can be defined using less than $d^2-1$ parameters makes it easier to find the optimal estimation procedure.

To summarize, in the multi-copy estimation problem we are considering here, $q_M(B|\rho)$ has the maximal information required to evaluate the figure of merit of the strategy described by the POVM  $M$. In other words,  if for two different strateges described by POVMs $M:\sigma(\Omega)\rightarrow\text{End}((\mathbb{C}^d)^{\otimes n})$ and $M':\sigma(\Omega)\rightarrow\text{End}((\mathbb{C}^d)^{\otimes n})$ it holds that
\begin{equation}
q_{M}\left(B|\rho\right)=q_{M'}\left(B|\rho\right) \label{eq:conditionals1}
\end{equation}
for all $B\in \sigma(\Omega)$ and $\rho\in \text{supp}(p)$ then they will have exactly the same performance in the estimation problem with respect to any figure of merit. On the other hand, $q_{M}(B|\vec{S}\in \vec{\Delta})$ has generally less information, i.e. it can be obtained by a coarse-graining of $q_{M}\left(B|\rho\right)$ but not necessarily vice versa.  However, in many reasonable figures of merit one does not need to specify $q_{M}\left(B|\rho\right)$ to specify the figure of merit of the measurement $M$; it is sufficient to specify $q_{M}(B|\vec{S}\in \vec{\Delta})$. If this is the case, then even if Eq.~\eqref{eq:conditionals1} doesn't hold, as long as the weaker constraint
\begin{equation}
q_{M}\left(B|\vec{S}\in \vec{\Delta}\right)=q_{M'}\left(B|\vec{S}\in \vec{\Delta}\right)\label{eq:conditionals2}
\end{equation}
holds for all $B\in \sigma(\Omega)$ and for all $l$-dimensional intervals $\vec{\Delta}$ that are assigned nonzero probability, the two strategies yield the same performance for the figure of merit of interest (See Fig. \ref{fig:estimation}). Eq.~\eqref{eq:conditionals2} states that learning the outcome of measurement $M$ is \emph{precisely as informative about the parameter $\vec{S}$} as learning the outcome of measurement $M'$.

An example of a common figures of merit, the average cost function, will be provided in Appendix~\ref{sec:costfunctions}.

\subsection{Main result}\label{sec:main}

\textbf{Scenario:} Suppose that  Alice randomly chooses an unknown state $\rho$ from the density operators in $\text{End}(\mathbb{C}^d)$  according to some probability density  $p$ (which we call the \emph{single-copy prior}) and sends $n$ qudits each prepared in the state $\rho$ to Bob through a quantum channel $\mathcal{E}:\text{End}((\mathbb{C}^d)^{\otimes n})\rightarrow \text{End}((\mathbb{C}^d)^{\otimes n})$.
Here, the density $p$ is defined relative to $d\rho$ a reference measure on the space of mixed states which is invariant under unitary transformations.\footnote{For example we can use the measure induced by the Hilbert-Schmidt inner product defined in \cite{Zyc}. }

Suppose that Bob makes measurements on the collection of $n$ systems.

Let parameters $\vec{\mathfrak{s}}(\cdot)=\left(\mathfrak{s}^{(1)}(\cdot),\cdots,\mathfrak{s}^{(l)}(\cdot)\right)$ be an arbitrary set of functions where $\mathfrak{s}^{{(i)}}: \text{supp}(p)\rightarrow \mathbb{R}$,  and  let $\vec{{S}}$ be the random variables defined as $\vec{{S}}\equiv \vec{\mathfrak{s}}(\rho) $ where $\rho$ is the random state Alice chooses.  We refer to $\vec{\mathfrak{s}}$ as the \emph{parameters}.
We say that the prior $p$ is invariant under a subgroup $H$ of $\text{U}(d)$, or equivalently, \emph{has $H$ as a symmetry} if
for all $\rho$ we have
\begin{equation}\label{fun-inv}
\forall V\in H: p(\rho)=p\left(V \rho V^{\dag}\right).
\end{equation}

We say that the parameter $\mathfrak{s}$ is invariant under a subgroup $H$ of $\text{U}(d)$, or equivalently, \emph{has $H$ as a symmetry} if for all $\rho \in \text{supp}(p)$, i.e. all $\rho$ assigned non-zero probability by the prior, we have
\begin{equation}\label{fun-inv2}
\forall V\in H:\vec{\mathfrak{s}}(\rho)=\vec{\mathfrak{s}}\left(V \rho V^{\dag}\right).
\end{equation}

We now present our main results, leaving the proofs to be presented in Sec.~\ref{sec:proofofmain}.  We begin with a version of the result where the assumptions are particularly simple.  These assumptions will be generalized shortly.

\begin{theorem}\label{Thm-main-chann}
Let $\mathcal{A}\subseteq \text{End}(\mathbb{C}^d)$ be a von Neumann algebra, and let $G_{\mathcal{A}}$ be the gauge group associated with it. Assume that:
\begin{enumerate}
\item \label{cond:symmetry} the prior $p$ and the vector of parameters $\vec{\mathfrak{s}}$ have the gauge group $G_{\mathcal{A}}$ as a symmetry;
\item \label{cond:identity} the channel $\mathcal{E}$ is the identity channel; 
\item \label{cond:pure} the prior $p$ has support only on the pure states.
\end{enumerate}
 Then for any given measurement with POVM $M:\sigma(\Omega)\rightarrow \text{End}((\mathbb{C}^{d})^{\otimes n})$, there is another measurement with POVM  $M':\sigma(\Omega)\rightarrow \text{End}((\mathbb{C}^{d})^{\otimes n})$ whose image is entirely confined to $\mathcal{A}^{\otimes n}$ (i.e., $M':\sigma(\Omega)\rightarrow \mathcal{A}^{\otimes n}$), such that $M'$ is as informative about $\vec{S}$ as $M$ is, i.e.,
\begin{equation}
  {q}_{M}\left(B|\vec{S}\in \vec{\Delta}\right)=q_{M'}\left(B|\vec{S}\in \vec{\Delta}\right)
\end{equation}
for all $B\in \sigma(\Omega)$ and all $l$-dimensional intervals $\vec{\Delta}$ which are assigned nonzero probability.
\end{theorem}

\begin{remark}\label{remark-main}
 An instance of the measurement described in theorem~\ref{Thm-main-chann} is $M'\equiv \mathcal{L}_{+}(M)$, where $\mathcal{L}_{+}$ is the unital quantum channel defined in Eq.~\eqref{M-tilde1}.
\end{remark}

One can generalize this theorem in two ways: from the identity channel to a class of nontrivial channels, and from a prior that has support only on pure states to a certain class of priors that have support on mixed states.  We begin by defining the classes in question.

We define a channel $\mathcal{E}$ to be \emph{noiseless on $\mathcal{A}^{\otimes n}$}  if for all states $\rho$ in $\text{End}((\mathbb{C}^d)^{\otimes n})$, $\mathcal{E}(\rho)$ and $\rho$ have the same reduction on the algebra $\mathcal{A}^{\otimes n}$, i.e.,
\begin{equation}
\forall R \in \mathcal{A}^{\otimes n}:\text{tr}(R \mathcal{E}(\rho))  = \text{tr}(R \rho),
\end{equation}
or equivalently, $\mathcal{T}_{G_{\mathcal{A}}}^{\otimes n}\circ \mathcal{E} = \mathcal{T}_{G_{\mathcal{A}} } ^{\otimes n}.$

Let prior density $\tilde{p}$ be one with support confined to the pure states.  Define a prior density $p$ to be a \emph{$\text{G}_{\mathcal{A}}$-distortion} of $\tilde{p}$ via channel $\mathcal{N}$ if it can be realized by sampling a pure state from $\tilde{p}$ and then applying a quantum channel $\mathcal{N}:\text{End}(\mathbb{C}^d)\rightarrow \text{End}(\mathbb{C}^d)$  to the state, where $\mathcal{N}$ is noiseless on $\mathcal{A}$ and is also $G_{\mathcal{A}}$-covariant i.e. $\forall V\in G_{\mathcal{A}}:\ \mathcal{N}(\cdot)=\mathcal{N}(V\cdot V^{\dag})$. (Recall that all these densities are defined relative to a fixed unitary invariant measure.)
We then have the following generalization of theorem  \ref{Thm-main-chann}.

\begin{theorem}\label{Thm-main-generalized} (Generalization of theorem \ref{Thm-main-chann})
 the implication in theorem \ref{Thm-main-chann} still holds if one weakens assumptions \ref{cond:identity} and \ref{cond:pure} to: \\
\ref{cond:identity}$'$. the channel $\mathcal{E}$ is noiseless on $\mathcal{A}^{\otimes n}$; \\
\ref{cond:pure}$'$. the prior $p$ is a $\text{G}_{\mathcal{A}}$-distortion of one that has support only on the pure states.
\end{theorem}

\begin{remark}\label{remark-main'}
Assume the prior $p$ is a $\text{G}_{\mathcal{A}}$-distortion of a prior over pure states via channel $\mathcal{N}$.  Then,  an instance of the measurement described in theorem~\ref{Thm-main-generalized} is $M'\equiv \mathcal{L}_{+}\circ(\mathcal{N}^{\dag})^{\otimes n}\circ \mathcal{E}^{\dag} (M)$, where $\mathcal{L}_{+}$ is the unital quantum channel defined in Eq.~\eqref{M-tilde1}.
\end{remark}

We now make explicit what our main theorem implies for multi-copy estimation problems.
\begin{corollary} \label{corollary:main}
Assume the figure of merit for a strategy $M$ in the $n$-copy estimation problem can be expressed as a functional of $q_{M}(B|\vec{S}\in \vec{\Delta})$ for some set of parameters $\vec{\mathfrak{s}}$. Then, if the assumptions of the theorem \ref{Thm-main-generalized} (or theorem \ref{Thm-main-chann}) hold for a von Neumann algebra $\mathcal{A}$, it follows that the POVM elements of the optimal measurement can be chosen to be in $\mathcal{A}^{\otimes n}$.
\end{corollary}

Corollary~\ref{corollary:main} implies that the optimal measurement has the gauge group G$_{\mathcal{A}}$ as a local symmetry.
Then, in the special case wherein the algebra $\mathcal{A}$ is commutative, by proposition~\ref{lem:com}, it follows that it can be implemented by measuring a set of observables which generates $\mathcal{A}$ separately on each of the $n$ qudits and then performing a classical processing on the outcomes.

To apply corollary~\ref{corollary:main}, the figure of merit for an estimation strategy $M$ must be a functional of  the conditional $q_{M}(B|\vec{S}\in \vec{\Delta})$.  In appendix~\ref{sec:costfunctions}, we demonstrate in an example how a common figure of merit, the expected cost for an arbitrary cost function, can be written in this form. 


We here describe an alternative way to state assumption \ref{cond:symmetry} of theorem \ref{Thm-main-chann} in the case where the prior $p$ has support only on pure states.

We begin with a definition.  We say that a function $g$ from states in $\text{End}(\mathbb{C}^d)$ to $ \mathbb{R}$ \emph{depends only on the reduction of the state to the algebra $\mathcal{A}$} if it can be expressed as
$$g(\rho)= f\left(\text{tr}(\rho \tilde{A}_{1}),\cdots, \text{tr}(\rho \tilde{A}_{D}) \right)$$
for some function $f:\mathbb{C}^{D}\rightarrow\mathbb{R}$, where $\{\tilde{A}_{1},\cdots,\tilde{A}_{D} \}\subset\mathcal{A}$ is a basis for $\mathcal{A}$.

In terms of this notion, the alternative statement of assumption \ref{cond:symmetry} is:

$1'.$ \emph{The prior $p$ and the vector of parameters $\vec{\mathfrak{s}}$ depend only on the reduction of the state to the algebra $\mathcal{A}$.}



The fact that assumption $1'$ implies assumption $1$ is clear: If $V\in G_{\mathcal{A}}$ then $\text{tr}(\rho V^{\dag} A V)=\text{tr}(\rho A)$ for arbitrary density operator $\rho$ in $\text{End}(\mathbb{C}^d)$ and arbitrary $A\in \mathcal{A}$. Then since according to assumption $1'$, $p$ and $\vec{\mathfrak{s}}$  can be expressed as a function of $\left(\text{tr}(\rho \tilde{A}_{1}),\cdots, \text{tr}(\rho \tilde{A}_{D})\right)$ we conclude that  $p(V\rho V^{\dag})=p(\rho)$ and $\vec{\mathfrak{s}}(V\rho V^{\dag})=\vec{\mathfrak{s}}(\rho)$ for arbitrary $\rho$ and arbitrary $V\in G_{\mathcal{A}}$.

The fact that assumption $1$ implies assumption $1'$ is true because of the following: Consider an arbitrary pair of pure states  $|\psi_{1}\rangle$  and $|\psi_{2}\rangle$ in the support of $p$. If for this pair of states there exists a unitary $V\in G_{\mathcal{A}}$ such that $V|\psi_{1}\rangle=|\psi_{2}\rangle$ then  assumption $1$ implies that 
$$\vec{\mathfrak{s}}(|\psi_{2}\rangle\langle\psi_{2}|)=\vec{\mathfrak{s}}\left(V |\psi_{1}\rangle\langle\psi_{1}| V^{\dag}\right)=\vec{\mathfrak{s}}\left( |\psi_{1}\rangle\langle\psi_{1}|\right)$$ 

On the other hand, if there does not exist a unitary $V\in G_{\mathcal{A}}$ such that $V|\psi_{1}\rangle=|\psi_{2}\rangle$ then $\vec{\mathfrak{s}}(|\psi_{2}\rangle\langle\psi_{2}|)$ could be different from $\vec{\mathfrak{s}}(|\psi_{1}\rangle\langle\psi_{1}|)$. In other words, to specify the value of $\vec{\mathfrak{s}}$ for a particular state $|\psi\rangle$ it is sufficient to  know the orbit of  $G_{\mathcal{A}}$ that $|\psi\rangle$  belongs to. From the results of \cite{Marvian} we know that there exists a unitary  $V\in G_{\mathcal{A}}$ for which $V|\psi_{1}\rangle=|\psi_{2}\rangle$ if and only if the reduction of two states $|\psi_{1}\rangle$  and $|\psi_{2}\rangle$ to the algebra $\mathcal{A}$ is the same, i.e. if $\langle\psi_{1}|\tilde{A_{i}}|\psi_{1}\rangle=\langle\psi_{2}|\tilde{A_{i}}|\psi_{2}\rangle$ for $\{\tilde{A}_{1},\cdots,\tilde{A}_{D} \}$ a basis of $\mathcal{A}$. This implies that by specifying the reduction of a state to the algebra one has enough information to infer the orbit that the state belongs to and so has enough information to find the value of $\vec{\mathfrak{s}}$. A similar argument can be applied for the density $p$. So, in general, if the prior $p$ is nonzero only on pure states, then any function which satisfies assumption $1'$ also satisfies assumption 1 and vice versa.



Note that the restriction to pure states plays an essential role in the equivalence of assumptions $1$ and $1'$ and this equivalence cannot be extended to the case of mixed states, i.e. in general for a parameter $\mathfrak{s}$ which satisfies assumption 1, $\mathfrak{s}(\rho)$ cannot be expressed as a function of $\text{tr}(\rho \tilde{A}_{1}),\dots,\text{tr}(\rho \tilde{A}_{D})$ if $\rho$ is mixed.
For instance, consider the case where $\mathcal{A}$ is the trivial algebra generated by the identity operator, so that $G_{\mathcal{A}}$ is the group of all unitaries on $\mathbb{C}^d$. In this case, the identity operator is a basis for $\mathcal{A}$ and consequently every state $\rho$ has the same reduction to $\mathcal{A}$. This means that the only functions that depend only on the reduction of the state to $\mathcal{A}$ are constant functions. However, there exist non-constant functions $\mathfrak{s}$, such as $\mathfrak{s}(\rho)=\text{tr}(\rho^{2})$, which are invariant under the group of all unitaries and therefore have the symmetry property required to satisfy assumption 1.  So the equivalence of assumption 1 and assumption $1'$ cannot be extended to the case of mixed states.




\subsection{Examples}

\subsubsection{Estimating parameters defined by a single observable}

A very simple example of a multi-copy estimation problem is the one considered by Hayashi \emph{et al.} \cite{HHH}. A pure state is chosen uniformly according to the Haar measure, and $n$ copies of the state are prepared.  The goal is to estimate the expectation value of an observable $A$ for the state.  Hayashi \emph{et al.} have shown that for a squared-error figure of merit, the optimal estimation scheme is to simply measure the observable $A$ separately on each system.  Our generalization of Schur-Weyl duality can be used to provide a very elementary proof of this result.  It can also be used to simplify the solution of estimation problems that are much more complicated, as we shall show.


Casting this in our language, the vector of parameters to be estimated, $\vec{\mathfrak{s}}(\rho)$, has only a single component, $\mathfrak{s}(\rho)=\text{tr}(A\rho)$. The figure of merit considered in Ref.~\cite{HHH} is the expected cost where the cost function is the squared error, i.e.
$$C(s_{est},\mathfrak{s}(\rho))= (s_{est} - \mathfrak{s}(\rho))^2.$$ Finally, the prior they consider is the unitarily-invariant measure over pure states and the channel $\mathcal{E}$ between the source and the estimator is the identity channel.  It follows that the assumptions of theorem~ \ref{Thm-main-chann} are all satisfied for the algebra $\mathcal{A}=\text{Alg}\{A,I\}$. Furthermore, one can show that the squared error for a measurement $M$ is a functional of the conditional $q_{M}(s_{est}\in\Delta_{est}|{S}\in {\Delta})$ in which $S$ is the actual value of the parameter, $s_{est}$ is the estimated value,  $\Delta$ and $\Delta_{est}$ are two arbitrary intervals in  $\mathbb{R}$ (see Appendix~\ref{sec:costfunctions}).  So the assumptions of proposition \ref{corollary:main} are satisfied. Consequently, the optimal measurement can be confined to $\mathcal{A}^{\otimes n}$, but given that $\mathcal{A}$ is commutative, it follows from corollary~\ref{lem:com} that it can be implemented by measuring the observable $A$ separately on each system and performing classical data processing on the outcomes.  So we have shown that the result of Hayashi \emph{et al.} is recovered as a special case of ours.

It is worth noting that for estimation problems involving only a single observable $A$ (or a set of commuting observables, which amounts to the same), there is in fact a very broad class of problems for which the optimal estimation can be achieved by separate measurements of $A$ on each system. Indeed, one can consider the estimation of any parameter that depends only on $A$, i.e. any function of the form $f(\text{tr}(\rho A),\text{tr}(\rho A^2),\text{tr}(\rho^2 A^2),\dots)$.  This includes the estimation of higher order moments of $A$, decisions about the sign of the expectation value of $A$, etcetera.  One can also take the prior $p$ to be arbitrary over pure states as long as it depends only on $A$. Also prior $p$ can be nonzero on mixed states as long as $p$ is a $G_{\mathcal{A}}$-distortion of a prior which is nonzero only on pure states.
Finally, there are many choices for the figure of merit.  We mention only two.  We could take the mutual information between the estimated values of the parameters and their actual values, or we could take the expected cost for an arbitrary cost function that depends only on $A$.  For all of these cases, the figure of merit for an estimation strategy $M$ is a functional of $q_{M}(B|\vec{S}\in \vec{\Delta})$ (see App.~\ref{sec:costfunctions}), so as long as the prior $p$ and the channel $\mathcal{E}$ satisfy assumptions $2'$ and $3'$ of theorem~\ref{Thm-main-generalized},  all the assumptions of corollary \ref{corollary:main} are satisfied, and separate measurements of $A$ suffice. Our result therefore constitutes a very significant generalization of the previously known results.

\subsubsection{Decision problem for a single qubit}
 Suppose we are given $n$ copies of qubit state $\rho$, a density operator in $\textrm{End}(\mathbb{C}^2)$.
For $b\in {0,1}$, define $$|\psi(\theta,b)\rangle\equiv\cos{\frac{\alpha_{b}}{2}} |0\rangle+e^{i\theta} \sin{\frac{\alpha_{b}}{2}}  |1\rangle $$ where $\alpha_{0}$ and $\alpha_{1}$ are distinct angles in the range $[0,\pi)$ and where $\theta \in [0,2\pi)$. Assume the single-copy prior $p(\rho)$ is as follows: the state is drawn from the set $\{ |\psi(\theta,b)\rangle \}$ where $\theta$ is uniformly distributed over $[0,2\pi)$ and $b$ has uniform distribution over $\{0,1\}$.  This prior is illustrated in Fig.~\ref{fig:BlochSpheres}(a).
 The goal is to get information about the value of the bit $b$ using $n$ copies of a state given according to this single-copy prior (this example is a decision problem). For instance, one might be interested to determine the value of the bit $b$ with minimum  probability of error. In general, we assume the goal is to generate an outcome in the outcome set $\Omega$ with $\sigma$-algebra $\sigma(\Omega)$ and the performance of different strategies are evaluated by a figure of merit which can be expressed as a functional acting on $q(B|b=b_{0})$, i.e. the probability of event $B\in\sigma(\Omega)$ while the value of $b$ is $b_{0}\in \{0,1\}$.
 
In this case, the parameter to be estimated is defined by $$\mathfrak{s}(|\psi(\theta,b)\rangle \langle \psi(\theta,b)|) = b.$$
 Adopting the convention that $|0\rangle$ and $|1\rangle$ are eigenstates of the Pauli observable $\sigma_z$, it is clear that the prior $p$ and the parameter to be estimated, $\mathfrak{s}$, are both invariant under unitaries of the form $e^{i\phi'}e^{i\phi \sigma_z}$ where $\phi,\phi' \in [0,2\pi)$, which describe phase shifts or rotations about the axis $\hat{z}$.  As we have seen in the section \ref{sec:charac} this group is a gauge group. 
 The algebra that corresponds to the commutant of this gauge group is $\mathcal{A} = \text{Alg}\{ \sigma_z, I \}$.  Finally, since the figure of merit depends only on $q(B|b=b_{0})$ the assumptions of corollary  \ref{corollary:main} are satisfied. [Note that since $\mathfrak{s}(|\psi(\theta,b)\rangle \langle \psi(\theta,b)|) = b$, $b$ can be thought as the random variable defined by parameter $\mathfrak{s}$ acting on states.]  Therefore, we can infer that to achieve the optimal estimation, it suffices to consider POVMs inside the algebra $\mathcal{A}^{\otimes n}$ and since $\mathcal{A}$ is commutative, it suffices to measure $\sigma_z$ on each system individually.  In other words, all the information we can get from the state $|\psi(\theta,b)\rangle^{\otimes n}$ about the value of $b$ we can also get from the mixed state  $[\cos^{2}{(\alpha_{b}) }|0\rangle\langle 0|+\sin^{2}{(\alpha_{b}) }|1\rangle\langle 1|]^{\otimes n}$.

\begin{figure} [h]
\begin{center}
\includegraphics[scale=0.75]{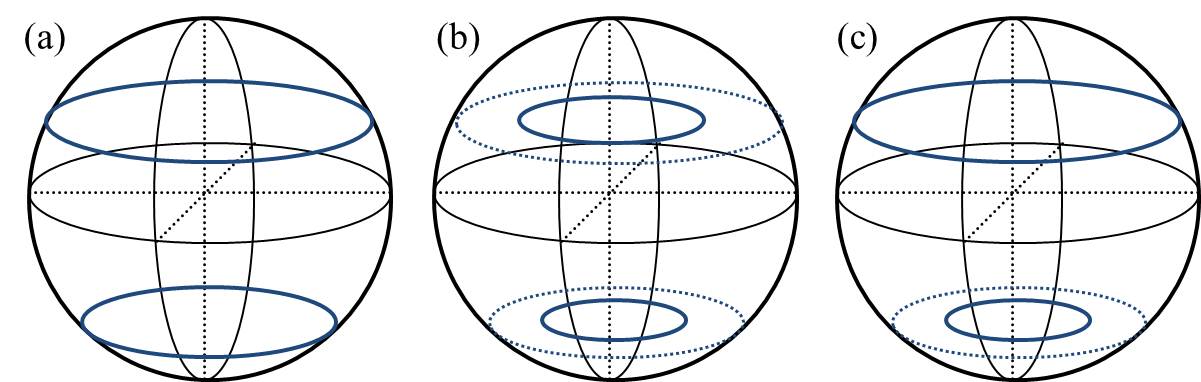}
\caption{The Bloch ball representation of the quantum states of a single qubit for three variations of a decision problem.  The pair of circles in each case indicate the support of the single-copy prior over states and the goal is to decide which circle the state is drawn from, given $n$ copies of the state. (a) A prior with support confined to  pure states. (b) A prior that is a gauge distortion of the first. (c) A prior for which unentangled measurement will not be generally sufficient to achieve optimal estimation.}\label{fig:BlochSpheres}
\end{center}
\end{figure}

Note, however, that if one acquires some information about $\theta$, then this information can be useful for estimating $b$: In the extreme case where we know the exact value of $\theta$, we can perform the Helstrom measurement \cite {Helstrom} for distinguishing the two pure states $|\psi(\theta,0)\rangle^{\otimes n}$ and $|\psi(\theta,1)\rangle^{\otimes n}$. So one estimation strategy is to use some of the qubits to estimate $\theta$ and then use this information to choose an optimal measurement for estimating $b$ using the rest of qubits. But our result shows that by this strategy one cannot get more information than what one gets by ignoring $\theta$ and measuring $\sigma_{z}$ on individual systems.
[Note that this result also implies that to get information about $\theta$ from each system we necessarily disturb its  information about $b$.  This can be interpreted as an example of information-disturbance tradeoff.]


\textit{Generalization to priors whose support is not confined to pure states.}
Theorem \ref{Thm-main-generalized}  implies that measuring $\sigma_z$ on each system is optimal even in the case where the single-copy prior is a $\text{G}_{\mathcal{A}}$-distortion of the one described above.  In this case a $\text{G}_{\mathcal{A}}$-distortion is implemented by 
a channel $\mathcal{N}$ that is covariant under phase shifts and noiseless on $\text{Alg}\{\sigma_z,I\}$.  The only channels having these properties are those corresponding to dephasing about the $\hat{z}$ axis (i.e. $\mathcal{N}(\rho)= (1-r)\rho+r\ \sigma_{z} \rho \sigma_{z}$ for $0<r<1$). For the single-copy prior that is achieved by this distortion, the state is drawn from the set $\{\rho(\theta,b) \equiv  \mathcal{N}(|\psi(\theta,b)\rangle \langle \psi(\theta,b)|)\}$ where $b$ and $\theta$ are distributed as before (for a given $b$, this describes a circle \emph{within} the Bloch ball). This prior is illustrated in Fig.~\ref{fig:BlochSpheres}(b). The parameter to be estimated is $\mathfrak{s}(\rho(\theta,b))=b$.  Note that both the prior and the parameter in this new estimation problem are invariant under the group of phase shifts.  Corollary  \ref{corollary:main} implies that the estimation problem so defined is also one wherein the optimal estimation is achieved by implementing a measurement of $\sigma_z$ on each qubit.

\textit{Example where unentangled measurements are generally not sufficient.}
Now suppose we are given $n$ copies of state  $\{ \rho(\theta,b) \}$  where  
$$\rho(\theta,0)= |\psi(\theta,0)\rangle \langle \psi(\theta,0)|\ \  \text{and}\ \ \rho(\theta,1)= \mathcal{N}(|\psi(\theta,1)\rangle \langle \psi(\theta,1)|)$$ where $\mathcal{N}$ is an arbitrary dephasing channel and where again $\theta$ is uniformly distributed between $(0,2\pi]$ and $b$ has arbitrary distribution. Effectively, we have a U(1)-orbit of pure states (a circle on the Bloch sphere) for $b=0$, and a dephased version of a distinct U(1)-orbit (a circle within the Bloch ball) for $b=1$.  This prior is illustrated in Fig.~\ref{fig:BlochSpheres}(c). Again the goal is to find the value of the bit $b$.

This estimation problem satisfies assumption 1 of theorem \ref{Thm-main-generalized} because the prior and the parameter have the same gauge group symmetry as the other examples considered in this section.  However, assumption 3' is not satisfied, and we can show that the optimal measurement is \emph{not} achieved by performing separate measurements of $\sigma_z$ on each qubit.

To see this, note first that 
because the $b=0$ states are pure while the $b=1$ states are mixed, the  purity of the state contains information about $b$.  Now consider the projective measurement which projects the state to the different irreps of $\mathcal{S}_{n}$ which show up in the representation $\textbf{P}(\mathcal{S}_{n})$  on $(\mathbb{C}^{d})^{\otimes n}$.  It is well known that this von Neumann measurement is highly nonlocal and requires interaction between all $n$ systems  \cite{Keyl-Werner}.  This projective measurement is  one that reveals information about the eigenvalues of the single-copy density operator and hence about its purity, as the following argument demonstrates. 

First, note that if the single-copy state is pure, then the $n$-copy state is in the symmetric subspace of $(\mathbb{C}^{d})^{\otimes n}$ and the outcome of the above projective measurement is fixed. On the other hand, if the single-copy state is mixed, then there is always a nonzero probability that the measurement projects the state to a subspace other than the symmetric subspace. In other words, there is a nonzero probability that the outcome of this measurement achieves an unambiguous discrimination between the mixed state case and the pure state case.  This implies that there is a nonzero probability of determining  the true value of $b$ unambgiuously. However, one can easily see that for the given prior by measuring $\sigma_{z}$ on each qubit it is not possible to unambigiously determine the true value of the bit $b$. Therefore, at least for some figures of merit, entangled measurements have advantage over unentangeled measurements.

 Incidentally, note that since the state of the total $n$ systems is a permutationally-invariant state, i.e. it commutes with $\textbf{P}(\mathcal{S}_{n})$ it is block diagonal in the irreps of $\mathcal{S}_{n}$ that show up in the representation of $\textbf{P}(\mathcal{S}_{n})$.  Therefore by performing the von Neumann measurement which projects into these blocks, the final state (forgetting the outcome of this measurement) will be the same as the initial state and therefore the statistics of any subsequent measurement will not be affected, that is, implementing such a measurement does not compromise the informativeness of any other measurement.

This phenomenon is generic.  In multi-copy decision problems in which the goal is to distinguish between a mixed state and a pure state, entangled measurements can achieve a better performance than unentangled measurements (at least with respect to some figures of merit).


\subsubsection{Decision problem for pair of qubits}
 \label{Ex-left-right} 
 In the previous example we assumed a bit is encoded in the state of one qubit and the goal is to acquire information about that bit using $n$ copies of that qubit state. Now suppose we modify the example in the following way: We assume each system consists of two qubits (rather than one),  left and right, i.e. the Hilbert space of each system is $\mathbb{C}^4\cong\mathcal{H}_{L}\otimes \mathcal{H}_{R}$ where $\mathcal{H}_{L/R}\cong\mathbb{C}^2$. Again, we are  given $n$ copies of state $\rho$ according to the single-copy prior $p(\rho)$ which is defined as follows:  the state is drawn from  the set 
$$\{ (\mathbb{I}\otimes V) |\psi(b) \rangle_{LR} \},$$ 
 where $b$ is uniformly distributed on $b\in \{0,1\}$,  $V$ is distributed according to the Haar measure over $\text{U}(2)$, and  
$$|\psi(0)\rangle_{LR} =|00\rangle_{LR}\ ,\ \ \ |\psi(1)\rangle_{LR} = \frac{|01\rangle_{LR} +|10\rangle_{LR}}{\sqrt{2}}.$$
The goal is again to get information about the bit $b$  and therefore, the parameter to be estimated is defined implicitly by the condition that   
$$\mathfrak{s}\left((\mathbb{I}\otimes V)|\psi(b) \rangle_{LR} \right)= b.$$
 It is then clear that the group of all unitaries acting on the right qubit, i.e. $\{\mathbb{I} \otimes V: V\in U(2)\}$ is a symmetry group of both the prior $p$ and the parameter $\mathfrak{s}$.  Moreover, this group of unitaries is clearly a gauge group, so it is a gauge symmetry of the prior and the parameter.  The algebra associated with this gauge group is the full algebra of operators on the left qubit, i.e. $\mathcal{A}\equiv \text{End}(\mathcal{H}_L) \otimes \mathbb{I}$.

 Again, we can see that for any figure of merit which depends only on $q(B|b=b_{0})$,  the assumptions of corollary  \ref{corollary:main} are satisfied and therefore to achieve the optimal estimation, it suffices to consider measurement operators inside the algebra $\mathcal{A}^{\otimes n}$.  It follows that it suffices to consider measurements that are nontrivial on the left qubits only. In other words, one can essentially ignore the right qubits. 
Note that deciding about the value of $b$ is also equivalent to deciding whether the reduced state of the right qubits is $(V|0\rangle)^{\otimes n}$  or $(\mathbb{I}/2)^{\otimes n}$. It follows that the $n$ right qubits do contain some information about the value of $b$, however, our results  imply that once one has the information contained in the left qubits, the information contained in the right qubits is redundant.


\subsection{Proof  of theorem  \ref{Thm-main-chann} and theorem \ref{Thm-main-generalized} } \label{sec:proofofmain}
To prove theorem \ref{Thm-main-chann} we first prove the following lemma which holds for any arbitrary subgroup of the unitary group.

\begin{lemma} (\textbf{From symmetry of the problem to symmetry of the measurement}) \label{lem:sym:main:thm}
In the scenario described in section~\ref{sec:main}, assume
the prior $p$ and the vector of parameters $\vec{\mathfrak{s}}$ are invariant under a subgroup $H$ of $\text{U}(d)$ which has the (normalized) Haar measure $d\mu$. Then for any measurement described by a POVM $M:\sigma(\Omega)\rightarrow\text{End}((\mathbb{C}^d)^{\otimes n})$, the measurement described by
$$\tilde{M}\equiv \mathcal{T}_{\textbf{Q}(H)}(M)= \int_{H} d\mu(V)\  V^{\otimes n} M {V^{\dag}}^{\otimes n}$$
 is as  informative as $M$ about $\vec{\mathfrak{s}}$,
that is,
\begin{equation}
  {q}_{M}\left(B|\vec{S} \in \vec{\Delta}\right)=q_{\tilde{M}}\left(B|\vec{S} \in \vec{\Delta}\right)
\end{equation}
for all $B\in \sigma(\Omega)$ and all $l$-dimensional intervals $\vec{\Delta} \subseteq \mathbb{R}^l$ which are assigned nonzero probability.
\end{lemma}
\begin{proof}
First note that for any $B\in\sigma(\Omega)$
$$q_{M}(B|\rho)=\text{tr}\left(\rho^{\otimes n} M(B)\right)$$
and
$$q_{\tilde{M}}(B|\rho)=\text{tr}\left(\rho^{\otimes n}  [\int_{H} d\mu(V)\  V^{\otimes n}  M(B) {V^{\dag}}^{\otimes n} ]\right)$$
Therefore, by the cyclic property of the trace,
\begin{equation} \label{proof:lem:cond}
q_{\tilde{M}}(B|\rho)=\int_{H} d\mu(V)\ q_{M}(B|V\rho V^{\dag})
\end{equation}
On the other hand,
\begin{equation}\label{Rel.Con}
{q}_{M}\left(B|\vec{S}\in \vec{\Delta}\right)=\frac{1}{\text{Pr}(\vec{S} \in \vec{\Delta})}\int_{\vec{{S}}\in \vec{\Delta}}  d\rho\ p(\rho) \ \ q_{M}(B|\rho)
\end{equation}
and similarly
\begin{equation}\label{Rel.Con}
{q}_{\tilde{M}}\left(B|\vec{{S}}\in \vec{\Delta}\right)=\frac{1}{\text{Pr}(\vec{S} \in \vec{\Delta})}\int_{\vec{{S}}\in \vec{\Delta}}  d\rho\ p(\rho) \ \ q_{\tilde{M}}(B|\rho)
\end{equation}
where $\text{Pr}(\vec{S} \in \vec{\Delta})$ is defined as
\begin{equation} \label{eq:marginS}
\text{Pr}(\vec{S} \in \vec{\Delta}) \equiv \int_{\vec{\mathfrak{s}}(\rho)\in \vec{\Delta}} \text{d}\rho\; p(\rho).
\end{equation}
But
\begin{align*}
 &\int_{\vec{{S}}\in \vec{\Delta}}  d\rho\ p(\rho) \ \ q_{\tilde{M}}(B|\rho)\\ &=
 \int_{\vec{{S}}\in \vec{\Delta}}  d\rho\ p(\rho)\  \int_{H} d\mu(V)\ q_{M}(B|V\rho V^{\dag})
\\ &=
 \int_{H} d\mu(V) \int_{\vec{\mathfrak{s}}(V\rho V^{\dag})\in \vec{\Delta}}  d\rho\ p(V\rho V^{\dag}) \ q_{M}(B|V\rho V^{\dag})
\\ &=
\int_{H} d\mu(V) \int_{\vec{{S}}\in \vec{\Delta}}  d\rho\ p(\rho) \ \ q_{M}(B|\rho)
\\ &=
\int_{\vec{{S}}\in \vec{\Delta}}  d\rho\ p(\rho) \ \ q_{M}(B|\rho)
\end{align*}
where to get the second line we use Eq.(\ref{proof:lem:cond}), to get   the third line we use the invariance of $p$ and $\vec{\mathfrak{s}}$ under $H$, to get the fourth line we use the fact that the measure $d\rho$ is invariant under unitary transformations  and to get the last line we use the fact that the Haar measure of $H$ is normalized.
This completes the proof.
\end{proof}

\begin{proof}(\textbf{Theorem} \ref{Thm-main-chann})

According to the first condition in theorem~\ref{Thm-main-chann}, the prior $p$ and the parameters $\vec{\mathfrak{s}}$ are invariant under the gauge group $G_{\mathcal{A}}$. So we can use lemma  \ref{lem:sym:main:thm} for the symmetry group $G_{\mathcal{A}}$. This implies that  for any given POVM $M:\sigma(\Omega)\rightarrow\text{End}((\mathbb{C}^d)^{\otimes n})$ and
\begin{equation}\label{proof:col:tw}
\tilde{M}\equiv \mathcal{T}_{\textbf{Q}(G_{\mathcal{A}})}(M) = \int_{G_{\mathcal{A}}} d\mu(V)\  V^{\otimes n} M {V^{\dag}}^{\otimes n}
\end{equation}
  it holds that
 \begin{equation} \label{proof:Glob:Sym}
{q}_{\tilde{M}}\left(B|\vec{S}\in \vec{\Delta}\right)={q}_{M}\left(B|\vec{S}\in \vec{\Delta}\right)
\end{equation}
for all $B\in \sigma(\Omega)$ and all $l$-dimensional intervals $\vec{\Delta}$ which are assigned nonzero probability. Now according to assumption 3 of theorem  \ref{Thm-main-chann}, the prior $p$ is nonzero only for pure states. So for all states in $\{\rho^{\otimes n}:\rho\in \text{supp}(p)\}$, i.e. the states Alice is sending to Bob, the support of the state is restricted to the symmetric subspace of $(\mathbb{C}^d)^{\otimes n}$. Since, by assumption 2, the channel is assumed to be the identity map, Bob receives the same state. Therefore all states that Bob receives are restricted to the symmetric subspace of $(\mathbb{C}^d)^{\otimes n}$.  By virtue of corollary \ref{Thm-Meas-LG Equiv}, this together with the fact that  the measurement $\tilde{M}$ has global symmetry imply
$\forall B\in \sigma(\Omega)$ and $\forall\rho\in\text{supp}(p)$
$$\text{tr}\left(\tilde{M}(B) \rho^{\otimes n}\right)=\text{tr}\left(\mathcal{L}_{+}(\tilde{M}(B)) \rho^{\otimes n}\right)$$

Define $M'\equiv\mathcal{L}_{+}(\tilde{M})$ where $\mathcal{L}_{+}$ is the superoperator defined in Eq.(\ref{M-tilde1}) of theorem \ref{pro-map-s}.   Then the above equality implies that
 \begin{equation}
{q}_{\tilde{M}}\left(B|\vec{{S}}\in \vec{\Delta}\right)={q}_{M'}\left(B|\vec{{S}}\in \vec{\Delta}\right)
\end{equation}
for all $B\in \sigma(\Omega)$ and all $\vec{\Delta}$ which are assigned nonzero probability. This together with  Eq.(\ref{proof:Glob:Sym}) implies  that for arbitrary POVM $M$
 \begin{equation}
{q}_{M}\left(B|\vec{{S}}\in \vec{\Delta}\right)={q}_{M'}\left(B|\vec{{S}}\in \vec{\Delta}\right)
\end{equation}
for all $B\in \sigma(\Omega)$ and all $\vec{\Delta}$ which are assigned nonzero probability.

Finally, using the fact that $\Pi_{+}$ commutes with $V^{\otimes n}$ for arbitrary $V\in \text{U}(d)$ we can easily see that
$$\mathcal{L}_{+}(\tilde{M})=\mathcal{L}_{+}({M}),$$
so that
$$M' = \mathcal{L}_{+}({M}).$$
From theorem \ref{pro-map-s}, we know that the image of $\mathcal{L}_{+}$ is in $\mathcal{A}^{\otimes n}$ and therefore so is $M'(B)$ for arbitrary $B\in \sigma(\Omega)$.
\end{proof}

\begin{proof} (\textbf{Theorem \ref{Thm-main-generalized}})

We first prove the special case of theorem~\ref{Thm-main-generalized} where assumptions
\ref{cond:symmetry}, \ref{cond:identity}' and  \ref{cond:pure} hold. In other words, we first prove the theorem for the case of general channels which satisfy the assumptions of theorem  \ref{Thm-main-generalized} but for the special case where the prior is still nonzero only on pure states.
 Then we extend the result to the case of general priors which satisfy the assumption \ref{cond:pure}'.

\noindent\textbf{(i) Generalization to non-identity channels, pure state priors:}

  The idea is to convert the estimation problem with  channel $\mathcal{E}$ to another estimation problem with the identity channel and then apply the result of theorem \ref{Thm-main-chann} to this new estimation problem.

For any estimation problem described  by the parameters $\vec{\mathfrak{s}}$, prior $p$, and the channel $\mathcal{E}$, we consider the two following scenarios:
\begin{itemize}
\item Scenario (a) in which Alice prepares $n$ copies of the state $\rho$  according to the probability density $p(\rho)$ and sends them through the channel $\mathcal{E}$ and  then Bob  performs a measurement described by  POVM $M: \sigma(\Omega)\rightarrow \text{End}((\mathbb{C}^d)^{\otimes n})$,  and
\item Scenario (b)  in which Alice prepares $n$ copies of the state $\rho$ according to the probability density $p(\rho)$ but then sends them  through the identity channel  and Bob performs the measurement described by  POVM $\mathcal{E}^{\dag}(M)$  on the systems.
\end{itemize}
The definitions of these two scenarios immediately imply
\begin{equation}\label{proof:Eq:cond}
{q}^{(a)}_{M}\left(B|\vec{{S}}\in \vec{\Delta}\right)=q^{(b)}_{\mathcal{E}^{\dag}(M)}\left(B|\vec{{S}}\in \vec{\Delta}\right)
\end{equation}
where the left and right hand sides describe the conditional for the scenarios $(a)$ and $(b)$ respectively.  This is true because in the scenario $(a)$ the probability of event  $B\in \sigma(\Omega)$ given that Alice has chosen state $\rho$ is $\text{tr}\left(M(B)\mathcal{E}(\rho^{\otimes n})\right)$. On the other hand, in the scenario $(b)$,
 the probability of event $B\in\sigma(\Omega)$ given that the state chosen by Alice is $\rho$ is $\text{tr}\left(\mathcal{E}^{\dag}(M(B))\rho^{\otimes n}\right)$. But since
$$\text{tr}\left(M(B)\mathcal{E}(\rho^{\otimes n})\right)=\text{tr}\left(\mathcal{E}^{\dag}(M(B))\rho^{\otimes n}\right)$$
for all $\rho\in \text{supp}(p)$ and $B\in \sigma(\Omega)$, Eq.(\ref{proof:Eq:cond}) follows.

Now in the scenario (b), where  the channel is the identity map,  we can apply  theorem~\ref{Thm-main-chann}. Note that the assumptions of this theorem are satisfied for the gauge group $G_{\mathcal{A}}$. This implies
\begin{equation}\label{proof:b:b:cond}
q^{(b)}_{\mathcal{L}_{+}(\mathcal{E}^{\dag}(M))}\left(B|\vec{{S}}\in \vec{\Delta}\right)
=q^{(b)}_{\mathcal{E}^{\dag}(M)}\left(B|\vec{{S}}\in \vec{\Delta}\right)
\end{equation}
Since the channel $\mathcal{E}$ is noiseless on $\mathcal{A}^{\otimes n}$  (assumption \ref{cond:identity}$'$) then all elements of $\mathcal{A}^{\otimes n}$ are fixed points of $\mathcal{E}^{\dag}$. (The fact that $\mathcal{E}$ is noiseless on $\mathcal{A}^{\otimes n}$ implies that for any operators $R_{1}\in \text{End}((\mathbb{C}^d)^{\otimes n})$ and $R_{2}\in\mathcal{A}^{\otimes n}$ it holds that $\text{tr}\left(R_{2}\mathcal{E}(R_{1})\right)=\text{tr}\left(R_{2}R_{1}\right)$. But this implies that   $\text{tr}\left(\mathcal{E}^{\dag}(R_{2})R_{1}\right)=\text{tr}\left(R_{2}R_{1}\right)$ which proves the claim.)

Then since elements of $\mathcal{A}^{\otimes n}$ are fixed points of $\mathcal{E^{\dag}}$ and since the image of  $\mathcal{L}_{+}$  is in
$\mathcal{A}^{\otimes n}$ we conclude that
$$\mathcal{E}^{\dag}\circ \mathcal{L}_{+}=\mathcal{L}_{+}$$
Putting this into Eq.(\ref{proof:b:b:cond}) we find
\begin{eqnarray*} \label{eq:kk1}
q^{(b)}_{ \mathcal{E}^{\dag}\circ\mathcal{L}_{+}\circ\mathcal{E}^{\dag}(M)}\left(B|\vec{{S}}\in \vec{\Delta}\right)=
q^{(b)}_{\mathcal{E}^{\dag}(M)}\left(B|\vec{{S}}\in \vec{\Delta}\right)
\end{eqnarray*}
Now for the conditionals on each side of this equality, we use Eq.(\ref{proof:Eq:cond}) to find the measurement in the scenario (a) that yields the same conditional.  We infer that
\begin{eqnarray} \label{eq:kk2}
q^{(a)}_{\mathcal{L}_{+}\circ\mathcal{E}^{\dag}(M)}\left(B|\vec{{S}}\in \vec{\Delta}\right)=
q^{(a)}_{M}\left(B|\vec{{S}}\in \vec{\Delta}\right),
\end{eqnarray}
and this holds for arbitrary $\rho\in \text{supp}(p)$ and event $B\in \sigma(\Omega)$ and arbitrary  POVM $M:\sigma(\Omega)\rightarrow \text{End}((\mathbb{C}^d)^{\otimes n})$. This completes the proof of the special case of the theorem where the prior $p$ is nonzero only for pure states. Note that in this particular case one can choose
$$M'\equiv \mathcal{L}_{+}\circ\mathcal{E}^{\dag}(M).$$

\noindent\textbf{(ii) Generalization to mixed state prior:}

According to assumption  1 the prior  $p$ is invariant under $G_{\mathcal{A}}$ and according to assumption 3', it can be realized by first sampling a pure state from $\tilde{p}$ and then applying channel $\mathcal{N}$ to the state where $\mathcal{N}$ is both $G_{\mathcal{A}}$ covariant and noiseless on $\mathcal{A}$.  Then one can easily see that the prior $\tilde{p}$ can always be chosen to be invariant under $G_{\mathcal{A}}$. In other words, for any given prior $\tilde{p}$ which satisfies the above properties there exists a prior $p'$ defined as
\begin{equation} \label{eq:pprime}
p'(\cdot)\equiv\int_{G_{\mathcal{A}}} d\mu(V)\ \tilde{p}(V\cdot V^{\dag})
\end{equation}
which also satisfies these properties, i.e. $p'$ is nonzero only on pure states and furthermore one can realize the prior $p$ by sampling a pure state from $p'$ and then applying the quantum channel $\mathcal{N}$ to the state. In addition to these properties, definition \ref{eq:pprime} guarantees that $p'$ is  also invariant under $G_{\mathcal{A}}$.

Now consider the estimation problem which is specified by the parameters $\vec{\mathfrak{s}}$, the prior $p$ and the channel $\mathcal{E}$ which satisfy all the assumptions of theorem \ref{Thm-main-generalized}.  We call this \emph{estimation problem (a)}.   Now define \emph{estimation problem (b)} via the following modifications of problem (a):
\begin{enumerate}
\item We change the prior $p$ to $p'$ defined in Eq.~\eqref{eq:pprime}.
\item We change the parameters $\vec{\mathfrak{s}}$ to $\vec{\mathfrak{s}}'$ where
\begin{equation} \label{eq:sprime}
\vec{\mathfrak{s}}'(\cdot)\equiv \vec{\mathfrak{s}}\left(\mathcal{N}(\cdot) \right)
\end{equation}
 and so naturally replace the random variables $\vec{S}$ induced by parameters $\vec{\mathfrak{s}}$ to the random variables $\vec{S}'$ induced by parameters $\vec{\mathfrak{s}}'$.
\item We change the channel $\mathcal{E}$ in the problem (a) to the channel
\begin{equation} \label{eq:Eprime}
\mathcal{E}'\equiv\mathcal{E}\circ {\mathcal{N}}^{\otimes n}.
\end{equation}
\end{enumerate}
For any POVM $M:\sigma(\Omega)\rightarrow \text{End}((\mathbb{C}^d)^{\otimes n})$ let
 $${q}^{(a)}_{M}\left(B|\vec{S}\in \vec{\Delta}\right)$$
be the conditional that in problem (a) an event $B\in \sigma(\Omega)$ happens given  $\vec{{S}}\in \vec{\Delta}$ and similarly
 $${q}^{(b)}_{M}\left(B|\vec{S}'\in \vec{\Delta}\right)$$
 be the conditional that in problem $(b)$ an event $B\in \sigma(\Omega)$ happens given  $\vec{{S}}'\in \vec{\Delta}$.

Now one can easily see that
by the manner in which they are defined, the parameters $\mathfrak{s}'$, prior $p'$ and channel $\mathcal{E}'$ of problem (b) satisfy  all the assumptions of the theorem.

On the other hand, since $p'$ is nonzero only for pure states then in the case of problem (b) we can use the result of part (i) of this proof, Eq.~\eqref{eq:kk2}, which implies that for any POVM $M: \sigma(\Omega)\rightarrow \text{End}((\mathbb{C}^d)^{\otimes n})$
\begin{equation} \label{proof:Eq:cond:BB}
{q}^{(b)}_{M}\left(B|\vec{S}'\in \vec{\Delta}\right)=q^{(b)}_{\mathcal{L}_{+}\circ\mathcal{E}'^{\dag}(M)}\left(B|\vec{S}'\in \vec{\Delta}\right)
\end{equation}
for all $B\in \sigma(\Omega)$ and for all $l$-dimensional intervals $\vec{\Delta}$ which are assigned nonzero probability.

Then it can be shown that for any POVM $M: \sigma(\Omega)\rightarrow \text{End}((\mathbb{C}^d)^{\otimes n})$  it holds that
\begin{equation} \label{proof:Eq:cond:AB}
{q}^{(a)}_{M}\left(B|\vec{{S}}\in \vec{\Delta}\right)=q^{(b)}_{M}\left(B|\vec{{S}}'\in \vec{\Delta}\right)
\end{equation}
for all $B\in \sigma(\Omega)$ and for all $l$-dimensional intervals $\vec{\Delta}$ which are assigned nonzero probability. We present the proof of this equality at the end. Now this equality allows us to transform the conditionals for problem (a) to the conditionals for the problem (b). Applying Eq.~\eqref{proof:Eq:cond:AB} to both sides of Eq.(\ref{proof:Eq:cond:BB}), we get
\begin{equation}
{q}^{(a)}_{M}\left(B|\vec{{S}}\in \vec{\Delta}\right)=q^{(a)}_{\mathcal{L}_{+}\circ\mathcal{E}'^{\dag}(M)}\left(B|\vec{{S}}\in \vec{\Delta}\right)
\end{equation}
Recall that the problem (a) is the original problem in the statement of theorem. So, defining
$${M}'\equiv \mathcal{L}_{+}\circ\mathcal{E}'^{\dag}(M)=\mathcal{L}_{+}\circ{\mathcal{N}^{\dag}}^{\otimes n}\circ\mathcal{E}^{\dag}(M)  $$
we conclude that in the original problem  for arbitrary POVM $M$, for arbitrary $B\in\sigma(\Omega)$ and for arbitrary $\vec{\Delta}$, it holds that
\begin{equation}
{q}_{M}\left(B|\vec{{S}}\in \vec{\Delta}\right)=q_{M'}\left(B|\vec{{S}}\in \vec{\Delta}\right)
\end{equation}
where for all $B\in\sigma(\Omega)$,  $M'(B)$ is in $\mathcal{A}^{\otimes n}$ as it is claimed in the theorem.

So it remains only to prove that Eq.(\ref{proof:Eq:cond:AB}) holds.
Let $\vec{\Delta} \subseteq \mathbb{R}^l$, and define
probability measures
\begin{align*}
&\text{Pr}^{(a)}\left(\vec{{S}}\in \vec{\Delta}\right)\equiv\int_{\vec{\mathfrak{s}}(\rho)\in \vec{\Delta}} d\rho\  p(\rho)
 \ \ \ \ \  \text{and},\\
&\text{Pr}^{(b)}\left(\vec{{S'}}\in \vec{\Delta}\right)\equiv\int_{\vec{\mathfrak{s}}'(\rho)\in \vec{\Delta}} d\rho\  p'(\rho).
\end{align*}

Note that
\begin{align*}
&q^{(a)}_{M}(B|\vec{{S}}\in \vec{\Delta})\equiv
\frac{\int_{\vec{{S}}\in \vec{\Delta}} d\rho  p(\rho) \  q^{(a)}_{M}(B|\rho)}{\text{Pr}^{(a)}\left(\vec{{S}}\in \vec{\Delta}\right)}\ \ \ \ \  \text{and},  \\
&q^{(b)}_{M}(B|\vec{{S}'}\in \vec{\Delta})\equiv
\frac{\int_{\vec{{S}'}\in \vec{\Delta}} d\rho  p'(\rho) \ q^{(b)}_{M}(B|\rho)}{\text{Pr}^{(b)}\left(\vec{{S}'}\in \vec{\Delta}\right)}
\end{align*}

Now using the definition $\mathfrak{s}'(\cdot) \equiv \mathfrak{s}(\mathcal{N}(\cdot))$ from Eq.~\eqref{eq:sprime},  we get

\begin{align}  \label{proof:prob:eq}
\nonumber \text{Pr}^{(b)}\left(\vec{{S'}}\in \vec{\Delta}\right)   &=\int_{\vec{\mathfrak{s}}(\mathcal{N}(\rho))\in \vec{\Delta}} d\rho\  p'(\rho)\\ \nonumber &=\int_{\vec{\mathfrak{s}}(\rho)\in \vec{\Delta}} d\rho\  p(\rho)\\ &=\text{Pr}^{(a)}\left(\vec{{S}}\in \vec{\Delta}\right) 
\end{align}
where to get the second line we have used the fact that by sampling a pure state from $p'$ and applying the channel $\mathcal{N}$ to it realizes the prior $p$.
Using exactly the same argument for
\begin{align*}
q^{(b)}_{M}(B|\rho)&=\text{tr}\left(\mathcal{E}'(\rho^{\otimes n}) M(B)\right)\ \ \ \ \ \ \ \ \ \  \text{and}\\
q^{(a)}_{M}(B|\rho)&=\text{tr}\left(\mathcal{E}(\rho^{\otimes n}) M(B)\right)
\end{align*}
and the definition $\mathcal{E}' \equiv \mathcal{E} \circ \mathcal{N}^{\otimes n}$, Eq.~\eqref{eq:Eprime}, we can prove that
\begin{align} \label{eq:llfinal}
\int_{\vec{{S}}\in \vec{\Delta}} d\rho \  p(\rho) \ q^{(a)}_{M}(B|\rho) = \int_{\vec{{S}'}\in \vec{\Delta}} d\rho\  p'(\rho) \  q^{(b)}_{M}(B|\rho)
\end{align}
Eqs.~\eqref{eq:llfinal} and  \eqref{proof:prob:eq}
together imply Eq.(\ref{proof:Eq:cond:AB}). This completes the proof.
\end{proof}

\section{Single-copy estimation problems for bipartite systems}\label{sec:bipartite}

Previously in this paper, the distinction between global and local symmetries was relative to the partitioning of the total system into $n$ copies of the system of interest. However, one can also consider estimation problems where the estimator gets only a single copy of the system of interest, and the distinction between global and local symmetries is relative to the partitioning of the system of interest into its components.  This case can be significantly different because the components of the system of interest need not correspond to copies of a single state.  Indeed, they could even be entangled.

In particular, we consider the case where the system has only \emph{two} components.  This case allows us to  obtain particularly strong constraints on the optimal measurement because the permutation group on two systems has only irreducible representations over the symmetric and antisymmetric subspaces and our duality only permits an inference from global symmetry to local symmetry within the symmetric and antisymmetric subspaces (as shown by the counterexample from Appendix \ref{app-count-ex}).

\subsection{General framework}
We begin with some notation.   The canonical representation of the permutation group on the pair is $\textbf{P}(\mathcal{S}_2) \equiv \{ \mathbb{I}_{d\times d}, \text{Swap} \}$, where $\mathbb{I}_{d\times d}$ is the identity operator on $(\mathbb{C}^d)^{\otimes 2}$ and $\text{Swap}$ is the unitary which exchanges the state of the two systems, i.e.  $\text{Swap}(|\psi\rangle|\phi\rangle) = |\phi\rangle |\psi\rangle$.  Under $\textbf{P}(\mathcal{S}_2)$, the space $(\mathbb{C}^{d})^{\otimes 2}$ decomposes as
\begin{equation}
(\mathbb{C}^{d})^{\otimes 2}\cong  [(\mathbb{C}^{d})^{\otimes 2}]_{+}\oplus [(\mathbb{C}^{d})^{\otimes 2}]_{-}
\end{equation}
Also, for any subgroup $H\subseteq \text{U}(d)$,  the collective representation of $H$ on the pair of systems is denoted by $\textbf{Q}(H)\equiv \{V^{\otimes 2}: V\in H\}$.

We are now in a position to state our result.

\textbf{Scenario:} Suppose that  Alice randomly chooses an unknown state $\rho$ from the density operators in $\text{End}\left((\mathbb{C}^d)^{\otimes 2}\right)$  according to some probability density  $p$ and sends a single system  in the state $\rho$ to Bob. Here, the density $p$ is defined relative to $d\rho$ a reference measure on the space of mixed states which is invariant under unitary transformations. Let $\vec{\mathfrak{s}}(\cdot)=\left(\mathfrak{s}^{(1)}(\cdot),\cdots,\mathfrak{s}^{(l)}(\cdot)\right)$ be an arbitrary set of functions where $\mathfrak{s}^{{(i)}}:\text{supp}(p)\rightarrow \mathbb{R}$,  and  let $\vec{{S}}$ be the random variables defined as $\vec{{S}}\equiv \vec{\mathfrak{s}}(\rho) $ where $\rho$ is the random state Alice chooses.

Recalling our earlier definitions, Eqs.~\eqref{fun-inv} and \eqref{fun-inv2}, of what it means for a prior $p$ and a vector of parameters $\vec{\mathfrak{s}}$ to have a symmetry, we can state our result as follows:

\begin{theorem}\label{Thm:bipartite}
Let $\mathcal{A}\subseteq \text{End}(\mathbb{C}^d)$ be a von Neumann algebra  with the gauge group $G_{\mathcal{A}}$. Assume that the prior $p$ and the vector of parameters $\vec{\mathfrak{s}}$
\begin{enumerate}
\item \label{cond:collgauge} have $\textbf{Q}(G_{\mathcal{A}})$ as a symmetry;
\item \label{cond:swap} have $\textbf{P}(\mathcal{S}_2)$ as a symmetry.
 \end{enumerate}
Then for any given measurement with POVM $M:\sigma(\Omega)\rightarrow \text{End}((\mathbb{C}^{d})^{\otimes 2})$ , there is another measurement whose POVM  is of the form
$$M'\equiv\Pi_{+}M_{+}\Pi_{+}+\Pi_{-}M_{-}\Pi_{-}$$
where
$M_{\pm}: \sigma(\Omega)\rightarrow\mathcal{A}^{\otimes 2}$ are POVMs,
such that $M'$ is as informative about $\vec{S}$ as $M$ is, i.e.,
\begin{equation}
  {q}_{M}\left(B|\vec{S}\in \vec{\Delta} \right)=q_{M'}\left(B|\vec{S}\in \vec{\Delta}\right)
\end{equation}
for all $B\in \sigma(\Omega)$ and all $l$-dimensional intervals $\vec{\Delta}$ which are assigned nonzero probability.
\end{theorem}

The proof is provided at the end of this section.  Note that unlike theorems~\ref{Thm-main-chann} and \ref{Thm-main-generalized}, the prior is not presumed to have support only on the pure states nor to be a gauge distortion of one that does.

\begin{remark}
The measurement $M'$ described in the above theorem can be implemented as follows: first perform the measurement which projects onto the symmetric/anti-symmetric subspace (the projective measurement described by projectors $\{\Pi_{+},\Pi_{-}\}$) and then, depending on the outcome of this measurement, perform either measurement $M_{+}$ or $M_{-}$ where both have local symmetry with respect to $G_{\mathcal{A}}$. The outcome of measurement $M'$ is the outcome of whichever of these measurements was performed.
\end{remark}

\subsection{Example}

Suppose that the prior over the pair of systems has support only on product states $\rho_{1}\otimes\rho_{2}$ where $\rho_{1},\rho_{2}\in \text{End}(\mathbb{C}^{d})$ and that it corresponds to choosing $\rho_{1}$ and $\rho_{2}$ independently according to a prior $p_{0}$, so that the joint prior has the form $p(\rho_{1}\otimes \rho_{2})=p_{{0}}(\rho_{1})p_{0}(\rho_{2})$.
Assume further that $p_0(\rho)$ only depends on the eigenvalues of $\rho$, so that $p_0(\cdot)=p_0\left(V(\cdot) V^{\dag}\right)$ for arbitrary $V\in \text{U}(d)$, i.e., $p_0$ has $\text{U}(d)$ as a symmetry.  It follows that the prior $p$ on the pair has $\textbf{Q}(\text{U}(d))$ as a symmetry, and consequently it also has $\textbf{Q}(H)$ as a symmetry for any subgroup $H$ of $\text{U}(d)$. Moreover, the prior $p$ is invariant under permutations, i.e. it has $\textbf{P}(\mathcal{S}_2)$ as a symmetry.

The goal is to estimate the parameter $\mathfrak{s}(\rho_{1}\otimes\rho_{2})= |\text{tr}\left(A \rho_{1}\right) -\text{tr}\left(A \rho_{2}\right)|$ for some observable $A$.  Let $\mathcal{A}$ denote the algebra generated by $\{\mathbb{I}_{d},A\}$ and let $G_{\mathcal{A}}$ denote the associated gauge group. It is clear that $\mathfrak{s}$ has $\textbf{Q}(G_{\mathcal{A}})$ as a symmetry.  Furthermore, $\mathfrak{s}$ is invariant under a swap of the pair of systems and therefore has $\textbf{P}(\mathcal{S}_2)$ as a symmetry as well.
The parameter $\mathfrak{s}$ therefore satisfies the assumptions of the above theorem for the gauge group $G_{\mathcal{A}}$.  Furthermore, because $G_{\mathcal{A}}$ is a subgroup of $\text{U}(d)$, the prior $p$ satisfies the assumptions of the above theorem as well.

So, for any figure of merit that can be defined as a functional acting on  ${q}_{M}\left(B|{S}={s}_{0}\right)$, the optimal estimation strategy corresponds to a POVM $M'$ of the form described in the theorem.  In our example, such a measurement has a particularly simple form.
First, note that because the two POVMs $M_{+}$ and $M_{-}$ have local symmetry with respect to $G_{\mathcal{A}}$ and because $\mathcal{A}$ is commutative, using proposition \ref{lem:com}, we can conclude that $M_{+}$ and $M_{-}$ can both be realized by measuring a Hermitian generator of $\mathcal{A}$ (e.g. the operator $A$) individually on each system and performing a classical processing of the outcome.  This means that in the case of this example, the POVM $M'$ described in the theorem can be realized by (i) performing the measurement which projects the state into the symmetric and antisymmetric subspaces, (ii) measuring the observable  $A$ individually on each system and (iii)  generating the outcome by a classical processing of the outcomes of these measurements.  So for all such $M'$s, the measurements are fixed and the part which is different is just the classical processing.

The same result holds for any other parameter which is invariant with respect to the exchange of the pair of systems and can be expressed in terms of an operator $A$,  such as $\mathfrak{s}(\rho_{1}\otimes\rho_{2})=\text{tr}\left(A \rho_{1}\right) +\text{tr}\left(A \rho_{2}\right)$ or more complicated parameters such as $\mathfrak{s}(\rho_{1}\otimes\rho_{2})=\text{tr}\left(A \rho^{k}_{1} \rho^{k}_{2}\right)+\text{tr}\left(A \rho^{k}_{2} \rho^{k}_{1}\right)$ for some integer $k$.

\subsection{Proof of theorem~\ref{Thm:bipartite}}

\begin{proof}(Theorem \ref{Thm:bipartite})
We need to apply lemma  \ref{lem:sym:main:thm} in its special case where $n=1$ and the Hilbert space of a single copy (which was denoted by $\mathbb{C}^{d}$ in the statement of lemma)   is   $\mathbb{C}^{d}\otimes \mathbb{C}^{d}$. The symmetry of the problem, denoted by $H \subseteq \text{U}(d^{2})$,  is the group generated by $\textbf{Q}(G_{\mathcal{A}})$ and $\textbf{P}(\mathcal{S}_2)$ together.
Then lemma \ref{lem:sym:main:thm} implies that  for any POVM $M: \sigma(\Omega)\rightarrow\text{End}(\mathbb{C}^{d}\otimes\mathbb{C}^{d})$ there is a POVM
$$\tilde{M}\equiv \mathcal{T}_{H}(M)= \int_{H} d\mu(V)\  V M V^{\dag}$$
 such that
$${q}_{\tilde{M}}\left(B|\vec{S}\in \vec{\Delta} \right)={q}_{M}\left(B|\vec{S}\in \vec{\Delta}\right)$$
for all $B\in \sigma(\Omega)$ and all $l$-dimensional intervals $\vec{\Delta}$ which are assigned nonzero probability. Now the above definition implies that  $\tilde{M}$ is invariant under permutation, i.e. $\tilde{M}=\text{Swap}[\tilde{M}] \text{Swap}$. This implies that
$$\tilde{M}=\Pi_{+}\tilde{M}\Pi_{+}+\Pi_{-}\tilde{M}\Pi_{-}.$$
$\tilde{M}$ also has global symmetry with respect to the gauge group $G_{\mathcal{A}}$, i.e. it commutes with $\textbf{Q}(G_{\mathcal{A}})$. Now corollary  \ref{Thm-Meas-LG Equiv} implies that for states whose supports are restricted to the symmetric/anti-symmetric subspaces a measurement with global symmetry with respect to gauge group $G_{A}$ can be simulated by a measurement whose POVM has local symmetry (and so its POVM elements are in $\mathcal{A}\otimes\mathcal{A}$). Therefore there exists POVMs $M_{+}$ and $M_{-}$ where
$M_{\pm}: \sigma(\Omega)\rightarrow\mathcal{A}\otimes \mathcal{A}$ such that
$$\Pi_{+}M_{+}\Pi_{+}=\Pi_{+}\tilde{M}\Pi_{+}\ \  \text{and}\ \ \Pi_{-}M_{-}\Pi_{-}=\Pi_{-}\tilde{M}\Pi_{-}$$
An example of $M_{\pm}$ is $\mathcal{L}_{\pm}(\tilde{M})$. Also since $\Pi_{\pm}\mathcal{L}_{\pm}(\tilde{M})\Pi_{\pm}=\Pi_{\pm}\mathcal{L}_{\pm}({M})\Pi_{\pm}$, it follows that $\mathcal{L}_{\pm}({M})$ is also an example of $M_{\pm}$. This completes the proof.
\end{proof}



\begin{acknowledgments}
We acknowledge helpful discussions with Giulio Chiribella.  Perimeter Institute is supported by the Government of Canada through Industry Canada and by the Province of Ontario through the Ministry of Research and Innovation. I. M. is supported by a Mike and Ophelia Lazaridis fellowship and NSERC.
\end{acknowledgments}

\appendix

\section{Proofs of lemma \ref{Lem-Per} and theorem  \ref{pro-map-s}} \label{Sec-Proof}

Throughout these proofs we use the superoperator  $\mathcal{T}_{\mathcal{S}_{n}}:\text{End}({(\mathbb{C}^d)^{\otimes n}})\rightarrow \text{End}({(\mathbb{C}^d)^{\otimes n}})$
\begin{equation}\label{symmetrizer}
\mathcal{T}_{\mathcal{S}_{n}}(\cdot)\equiv \frac{1}{n!}\sum_{s\in \mathcal{S}_{n}} \textbf{P}(s) (\cdot) \textbf{P}^{\dag}(s)
\end{equation}
which maps any operator in $\text{End}({(\mathbb{C}^d)^{\otimes n}})$ to its symmetrized version (under permutation).

\begin{proof}(\textbf{lemma \ref{Lem-Per}})
First note that $\text{Alg}\{\textbf{Q}(G)\}\subseteq \text{Alg}\{G^{\times n}\}$ and furthermore all elements of $\text{Alg}\{\textbf{Q}(G)\}$ are permutationally invariant. So $\text{Alg}\{\textbf{Q}(G)\}$ is included in the permutationally invariant subalgebra of  $\text{Alg}\{G^{\times n}\}$. In the following, we prove the converse inclusion.

We prove this by induction.  First we prove that for arbitrary $V_{0}\in G$, the subspace spanned by $\mathcal{T}_{\mathcal{S}_n}(V_{0}\otimes I^{\otimes (n-1)})$ is in $\text{Alg}\{\textbf{Q}(G)\}$. Then by induction we prove it is true for $\mathcal{T}_{\mathcal{S}_n}(V_{0}\otimes\cdots\otimes V_{n})$ for arbitrary $V_{i}\in G:i=1,\cdots ,n$ which proves the claim.

For arbitrary unitary $V_{0}\in G$, clearly $V_{0}+V_{0}^{\dag}$ and $i(V_{0}-V_{0}^{\dag})$ are both Hermitian operators which commute with $G'$ (the centralizer of $G$). Therefore, all  operators of the form  $V_{0}(\theta,\phi)\equiv \exp{[i\theta(V_{0}+V_{0}^{\dag})+\phi(V_{0}-V_{0}^{\dag}) ]}$,  for arbitrary real numbers $\theta$ and $\phi$ are unitary and commute with $G'$.  By virtue of being a gauge group, $G$ includes all unitaries which commute with $G'$, and it therefore follows that $V_{0}(\theta,\phi)\in G$. We can easily see that
\begin{equation}
\frac{1}{2}(\frac{\partial}{\partial \phi}-i\frac{\partial}{\partial \theta})_{|_{\theta=\phi=0}} V_{0}(\theta,\phi)=V_{0}
\end{equation}
 This implies that
 \begin{equation}
\frac{1}{2} (\frac{\partial}{\partial \phi}-i\frac{\partial}{\partial \theta})_{|_{\theta=\phi=0}} V_{0}^{\otimes n}(\theta,\phi)=\sum_{k} V_{0}^{(k)}
 \end{equation}
 where $V_{0}^{(k)}\equiv \mathbb{I}^{\otimes (k-1)}\otimes V_{0}\otimes \mathbb{I}^{\otimes(n-k)} $.
 This means that for arbitrary $V_{0}\in G$ 
 \begin{equation}\label{proof-inc-lem}
 \mathcal{T}_{\mathcal{S}_n}(V_{0}\otimes \mathbb{I}^{\otimes (n-1)})\in\text{Alg}\{\textbf{Q}(G)\}.
\end{equation}
Next we assume that
 $$\mathcal{T}_{\mathcal{S}_n}(V_{0}\otimes\cdots\otimes V_{k-1}\otimes \mathbb{I}^{\otimes (n-k)})$$
is in $\text{Alg}\{\textbf{Q}(G)\}$ for arbitrary $V_{i}\in{G}:i=0,\cdots, k-1$.  This together with Eq.(\ref{proof-inc-lem}) imply that for arbitrary $V_{k}\in G$
$$\mathcal{T}_{\mathcal{S}_n}(V_{0}\otimes\cdots V_{k-1}\otimes \mathbb{I}^{\otimes (n-k)})\mathcal{T}_{\mathcal{S}_n}(V_{k}\otimes \mathbb{I}^{\otimes (n-1)})$$
is in $\text{Alg}\{\textbf{Q}(G)\}$. Expanding this, one can easily see that it can be written as
\begin{align*}
c_{1}\mathcal{T}_{\mathcal{S}_n}(V_{0}\otimes\cdots V_{k-1}\otimes V_{k}&\otimes \mathbb{I}^{\otimes (n-k-1)})+c_{2}\mathcal{T}_{\mathcal{S}_n}(U_{0}\otimes\cdots U_{k-1}{\otimes}\mathbb{I}^{\otimes (n-k)})
\end{align*}
for some nonzero coefficients  $c_{1},c_{2}$ and unitaries  $U_{i}\in{G}:i=0,\cdots, k-1$. Now since the sum and the second term each are in the span of $\text{Alg}\{\textbf{Q}(G)\}$ then we conclude that the first term is also in $\text{Alg}\{\textbf{Q}(G)\}$.  Note that  $k$ and $V_{i}\in G: i=0\cdots k$ are arbitrary. So by induction we have the lemma.
\end{proof}

\begin{proof} \textbf{(theorem  \ref{pro-map-s})}

Suppose for operator $M\in\text{End}\left((\mathbb{C}^{d})^{\otimes n}\right)$ it holds that $\Pi_{\pm} M\Pi_{\pm}$ commutes with $\textbf{Q}(G_{\mathcal{A}})$, i.e. 
\begin{equation}
\forall V\in G_{\mathcal{A}}: \Pi_{\pm}M\Pi_{\pm} \textbf{Q}(V)= \textbf{Q}(V) \Pi_{\pm}M \Pi_{\pm} 
\end{equation}
 Since $V^{\otimes n}$ commutes with $\Pi_{\pm}$ this implies
 \begin{equation} \label{commute}
\Pi_{\pm} M\Pi_{\pm} \textbf{Q}(V)  \Pi_{\pm}=\Pi_{\pm}  \textbf{Q}(V)  \Pi_{\pm} M \Pi_{\pm}
\end{equation}
This holds for arbitrary $V\in G_{\mathcal{A}}$. So we can conclude that for any operator $X$ in $\text{Alg}\{\textbf{Q}(G_\mathcal{A})\}$ we have
\begin{equation} \label{commute4}
\Pi_{\pm} M\Pi_{\pm} X  \Pi_{\pm}=\Pi_{\pm}  X  \Pi_{\pm} M \Pi_{\pm}
\end{equation}
According to  lemma \ref{Lem-Per},  $\text{Alg}\{\textbf{Q}(G_\mathcal{A})\}$ is equal to the span of the permutationally invariant subspace of  ${G^{\times n}_\mathcal{A}}$. Consider $V_{1}\otimes\cdots \otimes V_{n} $ an arbitrary element of  ${G^{\times n}_\mathcal{A}}$. 
Since $\mathcal{T}_{\mathcal{S}_n}(V_{1}\otimes\cdots\otimes V_{n})$ is in the permutationally invariant subspace of the span of ${G^{\times n}_\mathcal{A}}$,  it satisfies  Eq.(\ref{commute4}) and so

\begin{align} \label{commute8}
\Pi_{\pm} M\Pi_{\pm} &\left[\mathcal{T}_{\mathcal{S}_n}(V_{1}\otimes\cdots\otimes V_{n})\right]\Pi_{\pm}=
\Pi_{\pm}  \left[\mathcal{T}_{\mathcal{S}_n}(V_{1}\otimes\cdots\otimes V_{n}) \right]\Pi_{\pm} M \Pi_{\pm}
\end{align}

For arbitrary permutation $s\in {\mathcal{S}_{n}}$, $\textbf{P}(s)\Pi_{\pm}=\Pi_{\pm}\textbf{P}(s)=\eta \Pi_{\pm} $  for some $\eta\in \{\pm 1\}$. Therefore Eq.(\ref{commute8}) implies 
\begin{align*}
\Pi_{\pm} M\Pi_{\pm}  &\left[V_{1}\otimes\cdots \otimes V_{n} \right]  \Pi_{\pm}=\Pi_{\pm} \left[V_{1}\otimes\cdots \otimes V_{n} \right]\Pi_{\pm} M\Pi_{\pm}
\end{align*}
We multiply by $ [V_{1}^{\dag}\otimes\cdots \otimes V^{\dag}_{n}]\Pi_{\pm}$  on the right on both sides  of the above equality to obtain
\begin{align*}
 \Pi&_{\pm} M \Pi_{\pm}\left[(V_{1}\otimes\cdots \otimes V_{n})\Pi_{\pm} (V_{1}^{\dag}\otimes\cdots\otimes V^{\dag}_{n})\right]\Pi_{\pm}\\  &\ \ \ \ = \Pi_{\pm} \left[(V_{1}\otimes\cdots\otimes V_{n})\Pi_{\pm} M\Pi_{\pm} (V_{1}^{\dag}\otimes\cdots\otimes V^{\dag}_{n})\right]\Pi_{\pm}
\end{align*}
Now suppose on both sides we integrate over all elements of ${G^{\times n}_\mathcal{A}}$ using the Haar measure. Then the above equality implies
\begin{equation} \label{commut3}
\Pi_{\pm}  M\Pi_{\pm}[\mathcal{T}_{G_{\mathcal{A}}}^{\otimes n} (\Pi_{\pm})]\Pi_{\pm}=\Pi_{\pm}\mathcal{T}_{G_{\mathcal{A}}}^{\otimes n}  (\Pi_{\pm} M\Pi_{\pm} )\Pi_{\pm}
\end{equation}
Now we demonstrate how one can write $\Pi_{\pm}  M\Pi_{\pm}$ as $\Pi_{\pm}\mathcal{T}_{G_{\mathcal{A}}}^{\otimes n}  (\Pi_{\pm} M\Pi_{\pm} )\Pi_{\pm}$ times the inverse of $\Pi_{\pm}[\mathcal{T}_{G_{\mathcal{A}}}^{\otimes n} (\Pi_{\pm})]\Pi_{\pm}$.

Consider $\mathcal{T}_{G_{\mathcal{A}}}^{\otimes n} (\Pi_{\pm})$ and $\mathcal{T}_{G_{\mathcal{A}}}^{\otimes n}  (\Pi_{\pm} M\Pi_{\pm} )$ on the left and right hand sides of the above equality. First of all, since $\Pi_{\pm}$ and $\Pi_{\pm} M\Pi_{\pm} $ are both permutationally invariant then both  $\mathcal{T}_{G_{\mathcal{A}}}^{\otimes n} (\Pi_{\pm})$ and $\mathcal{T}_{G_{\mathcal{A}}}^{\otimes n}  (\Pi_{\pm} M\Pi_{\pm} )$ are permutationally invariant. Furthermore, since these two operators also commute with $G_\mathcal{A}^{\times n}$  then corollary \ref{lem-twr} implies that they are both in $\text{Alg}\{\textbf{Q}(G_\mathcal{A}')\}$. Second, since $\Pi_{\pm}$ commutes  with $\textbf{Q}(G_\mathcal{A}')$ in the case of $\mathcal{T}_{G_{\mathcal{A}}}^{\otimes n}  (\Pi_{\pm})$  we have another symmetry: $\mathcal{T}_{G_{\mathcal{A}}}^{\otimes n}  (\Pi_{\pm})$      commutes with $\textbf{Q}(G_\mathcal{A}')$. Considering this fact together with the fact that  $\mathcal{T}_{G_{\mathcal{A}}}^{\otimes n}  (\Pi_{\pm})$  is in $\text{Alg}\{\textbf{Q}(G_\mathcal{A}')\}$ we conclude that  $\mathcal{T}_{G_{\mathcal{A}}}^{\otimes n}  (\Pi_{\pm})$ should have the following form
\begin{equation}\label{decom-sym}
\mathcal{T}_{G_{\mathcal{A}}}^{\otimes n}  (\Pi_{\pm})=\bigoplus_{\mu} p_{\mu,\pm}\  P_{\mu}
\end{equation}
where $\mu$ labels all the irreps of $G_{\mathcal{A}}'$ which shows up in the representation $\textbf{Q}(G_\mathcal{A}')$  and $P_{\mu}$ is the projector to these irreps  and by virtue of $\mathcal{T}_{G_{\mathcal{A}}}^{\otimes n} $ being a completely positive map, all $p_{\mu,\pm}$'s are non-negative. Let $\Gamma_{\pm}$ be the set of all irreps of $G'_{\mathcal{A}}$ for which $p_{\mu},{\pm}$ is nonzero. So we can write Eq.(\ref{commut3}) as
\begin{equation}\label{commut5}
\Pi_{\pm}  M\Pi_{\pm}\left(\bigoplus_{\mu\in\Gamma_{\pm}} p_{\mu,\pm}\ P_{\mu}\right)\Pi_{\pm}=\Pi_{\pm}\left[\mathcal{T}_{G_{\mathcal{A}}}^{\otimes n}  (\Pi_{\pm} M\Pi_{\pm} )\right]\Pi_{\pm}
\end{equation}
Now consider the inverse of $\mathcal{T}_{G_{\mathcal{A}}}^{\otimes n}  (\Pi_{\pm})=\bigoplus_{\mu\in \Gamma_{\pm}} (p_{\mu,\pm}\  P_{\mu})$ on its support,  i.e., the operator
$$\bigoplus_{\mu\in \Gamma_{\pm}} p^{-1}_{\mu,\pm}\ P_{\mu}$$
By multiplying both sides of Eq.(\ref{commut5}) on the right with this operator and using the facts that
\begin{enumerate}
\item  $\Pi_{\pm}$ commutes with $\textbf{Q}(G_\mathcal{A}')$ and so it commutes with all $P_{\mu}$'s,
\item \begin{equation}\label{app:comm}
\Pi_{\pm} \left(\bigoplus_{\mu\in\Gamma_{\pm}} P_{\mu}\right)=\left(\bigoplus_{\mu\in\Gamma_{\pm}} P_{\mu}\right)\Pi_{\pm}=\Pi_{\pm}
\end{equation}
which is true because all $P_{\mu}$'s commute with $\Pi_{\pm}$ and   Eq.(\ref{decom-sym}) implies that the support of $\Pi_{\pm}$ is a subspace of the support of $\bigoplus_{\mu\in \Gamma_{\pm}} P_{\mu}$     and
\item  $\forall\mu:\ P_{\mu}\mathcal{T}_{G_{\mathcal{A}}}^{\otimes n}  (\Pi_{\pm} M\Pi_{\pm} )=\mathcal{T}_{G_{\mathcal{A}}}^{\otimes n}  (\Pi_{\pm} M\Pi_{\pm} )P_{\mu}$,
which is true because $\mathcal{T}_{G_{\mathcal{A}}}^{\otimes n}  (\Pi_{\pm} M\Pi_{\pm} )$ is  in the  span of $\textbf{Q}(G_\mathcal{A}')$
\end{enumerate}
 we get
\begin{equation}
\Pi_{\pm} M\Pi_{\pm}=\Pi_{\pm}\left( \bigoplus_{\mu\in\Gamma_{\pm}}  p^{-1}_{\mu,\pm}\ P_{\mu}\left[\mathcal{T}_{G_{\mathcal{A}}}^{\otimes n}  (\Pi_{\pm} M\Pi_{\pm} )\right]P_{\mu}\right) \Pi_{\pm}
\end{equation}
Therefore, defining $\Phi_{\pm}$ as
\begin{equation}
\Phi_{\pm}{(\cdot)}\equiv  \bigoplus_{\mu\in \Gamma_{\pm}}  p^{-1}_{\mu,\pm} \ P_{\mu}[\mathcal{T}_{G_{\mathcal{A}}}^{\otimes n}  (\Pi_{\pm} (\cdot)\Pi_{\pm} )]P_{\mu}
\end{equation}
we infer that
\begin{equation}\label{sup:phi+}
\Pi_{\pm} M\Pi_{\pm}=\Pi_{\pm}\Phi_{\pm}({M}) \Pi_{\pm}
\end{equation}
Because all $P_{\mu}$'s and $\mathcal{T}_{G_{\mathcal{A}}}^{\otimes n}  (\Pi_{\pm} (\cdot)\Pi_{\pm} )$ are in $\text{Alg}\{\textbf{Q}(G_{\mathcal{A}}')\}$, the image of $\Phi_{\pm}$ is as well. Note that since $G'_{\mathcal{A}}\subset \mathcal{A}$ this means that  the image of $\Phi_{\pm}$ is in the permutationally invariant subalgebra of $\mathcal{A}^{\otimes n}$.
Now defining $\mathcal{L}_{\pm}$ in terms of $\Phi_{\pm}$ via
$$\mathcal{L}_{\pm}{(\cdot)}\equiv \Phi_{\pm}(\cdot) + \left[\mathbb{I}^{\otimes n} - \Phi_{\pm}(\mathbb{I}^{\otimes n})\right]\mathrm{tr}(\cdot)/d^n\,$$
 we can infer the same properties for $\mathcal{L}_{\pm}$.  First note that
$$ \Phi_{\pm}(\mathbb{I}^{\otimes n})=\bigoplus_{\mu\in \Gamma_{\pm}} P_{\mu}$$
which together with Eq.(\ref{app:comm})  implies that
$\Pi_{\pm}[\mathbb{I}^{\otimes n} - \Phi_{\pm}(\mathbb{I}^{\otimes n})]\Pi_{\pm} =0$. This together with Eq.(\ref{sup:phi+}) and definition of $\mathcal{L}_{\pm}$  implies
\begin{equation}
\Pi_{\pm} M\Pi_{\pm}=\Pi_{\pm}\mathcal{L}_{\pm}({M}) \Pi_{\pm}\,,
\end{equation}
which is the third claim of theorem~\ref{pro-map-s}. Furthermore since the image of $\Phi_{\pm}$  is in the permutationally invariant subalgebra of $\mathcal{A}^{\otimes n}$ and since $\mathcal{A}$, being a von-Neumann algebra, includes identity, it follows that $\mathbb{I}^{\otimes n} - \Phi_{\pm}(\mathbb{I}^{\otimes n})$ is in the permutationally invariant subalgebra of $\mathcal{A}^{\otimes n}$.  This implies that the image of $\mathcal{L}_{\pm}(\cdot)$ is in this subalgebra, which is the second claim of theorem~\ref{pro-map-s}.

Furthermore, noting that  $\mathcal{T}_{G_{\mathcal{A}}}^{\otimes n}$  is completely positive  and the $p^{-1}_{\mu,\pm}$'s are all positive numbers we can conclude that $\Phi_{\pm}$ as a combination of completely positive maps is completely positive. This together with the fact that $\mathbb{I}^{\otimes n} - \Phi_{\pm}(\mathbb{I}^{\otimes n})$ is a projector (and so a positive operator) implies that $\mathcal{L}_{\pm}$ is completely positive. Finally,   it is straightforward to verify that $\mathcal{L}_{\pm}(\mathbb{I}^{\otimes n}) =\mathbb{I}^{\otimes n}$, so that it is unital which proves  the first claim of theorem~\ref{pro-map-s}.
\end{proof}

\section{Global symmetry with respect to non-gauge groups} \label{app-count-ex0}

We demonstrate here that a group that does not have the gauge property does not yield a dual reductive pair in the manner specified by theorems~\ref{cor-pairs}. That is, we present an example for a non-gauge group  $H\subseteq \text{U}(d)$ for which  the commutant of the algebra spanned by $\textbf{Q}(H)$ in $\text{End}((\mathbb{C}^d)^{\otimes n})$ is larger than the algebra  spanned by $\langle (H')^{\times n},  \textbf{P}(\mathcal{S}_{n})\rangle$. (Recall that for any group $H\subseteq \text{U}(d)$ it always holds that $\text{Alg}\{ (H')^{\times n},  \textbf{P}(\mathcal{S}_{n})\}\subseteq \text{Comm}\{\textbf{Q}(H)\}$).

As a simple example, consider $d=3$, $n=2$ where the group $H$ is the $j=1$ irreducible representation of $\text{SU}(2)$ which is a subgroup of $\text{U}(3)$. This group is not a gauge group: Schur's lemma implies that $H'=\{e^{i\theta}\mathbb{I}\}$ where $\theta \in (0,2\pi]$ and $\mathbb{I}$ is identity on $\mathbb{C}^{3}$ and so $H''=\text{U}(3)\neq H$.

Since $H'=\{e^{i\theta}\mathbb{I}\}$ then
\begin{align*}
\text{Alg}\{ (H')^{\times 2},  \textbf{P}(\mathcal{S}_{2})\}&=\text{Alg}\{  \textbf{P}(\mathcal{S}_{2})\}\\&=\{c_{+}\Pi_{+}+c_{-}\Pi_{-}:c_{\pm}\in \mathbb{C}\}
\end{align*}
where $\Pi_{+}$ and $\Pi_{-}$ are respectively the projectors to the symmetric and anti-symmetric subspace of $(\mathbb{C}^{3})^{\otimes 2}$. On the other hand, one can easily see that $\text{Comm}\{\textbf{Q}(H)\}$, the algebra of operators  commuting with $\textbf{Q}(H)$, is
$$\{c_{0}P_{j=0}+c_{1}P_{j=1}+c_{2}P_{j=2}, c_{0,1,2} \in \mathbb{C}\}$$
where $P_{j}$ is the projector to the subspace of $(\mathbb{C}^{3})^{\otimes 2}$ with total angular momentum $j$. Therefore the algebra of operators  commuting with $\textbf{Q}(H)$ is larger than  $\text{Alg}\{ (H')^{\times 2},  \textbf{P}(\mathcal{S}_{2})\}$\footnote{One can show that $P_{j=1}=\Pi_{-}$, in other words in this space any anti-symmetric state has the total angular momentum $j=1$ and any state with total angular momentum $j=1$ is anti-symmetric. This  implies that $P_{j=0}+P_{j=2}=\Pi_{+}$.}.

\section{Lack of duality outside the symmetric and antisymmetric subspaces}
\label{app-count-ex}

Here, we show that the restriction to the symmetric and anti-symmetric subspaces plays an essential role in theorem \ref{Thm-sym}  and the other results of section \ref{sec:restric}. Recall that theorem \ref{Thm-sym}   implies that for symmetric and anti-symmetric subspaces of $(\mathbb{C}^{d})^{\otimes n}$, denoted by $\left[(\mathbb{C}^{d})^{\otimes n} \right]_{\pm}$, and for any gauge group $G\subseteq \text{U}(d)$ it holds that
\begin{align*}
\text{Alg}&\{\Pi_{\pm}\textbf{Q}(G')\Pi_{\pm}\}=\text{Comm}\{\Pi_{\pm}\textbf{Q}(G)\Pi_{\pm}\}
\end{align*}
Let  $\lambda$ labels different irreps of the permutation group and  $\Pi_{\lambda}$ be the projector to the  subspace $\left[(\mathbb{C}^{d})^{\otimes n}\right]_{\lambda}$ of $(\mathbb{C}^{d})^{\otimes n}$ in which the representation $\textbf{P}(\mathcal{S}_{n})$ acts like the irrep  $\lambda$ of $\mathcal{S}_{n}$. The goal is to see whether in theorem \ref{Thm-sym}, or equivalently in the above equation, one can substitute the projection to the symmetric (anti-symmetric) subspace by the projection to an arbitrary irrep $\lambda$  of $\mathcal{S}_{n}$.

Clearly for any other irrep $\lambda$ of $\mathcal{S}_{n}$ it holds that
\begin{align*}
\text{Alg}&\{\Pi_{\lambda}\textbf{Q}(G')\Pi_{\lambda}\}\subseteq \text{Comm}\{\Pi_{\lambda}\textbf{Q}(G)\Pi_{\lambda}\}
\end{align*}
where by $\text{Comm}\{\Pi_{\lambda}\textbf{Q}(G)\Pi_{\lambda}\}$ we mean the commutant of $\Pi_{\lambda}\textbf{Q}(G)\Pi_{\lambda}$ in $\text{End}(\left[(\mathbb{C}^{d})^{\otimes n} \right]_{\lambda})$.

However, for subspaces other than symmetric and anti-symmetric subspaces, the elements of $\text{Comm}\{\Pi_{\lambda}\textbf{Q}(G)\Pi_{\lambda}\}$ are not necessarily permutationally invariant while $\text{Alg}\{\Pi_{\lambda}\textbf{Q}(G')\Pi_{\lambda}\}$ is permutationally invariant. So to generalize theorem \ref{Thm-sym} to other representations of the permutation group one should make an extra assumption to guarantee that  the elements of both sides are  permutationally  invariant.  Then one may expect the following to be true:
 a natural generalization of theorem \ref{Thm-sym} will be 
\begin{align*}
\text{Alg}&\{\Pi_{\lambda}\textbf{Q}(G')\Pi_{\lambda}\}= \text{Comm}\{\Pi_{\pm}\textbf{Q}(G)\Pi_{\pm}\}\ \cap \ \text{Comm}\{\Pi_{\lambda}\textbf{P}(\mathcal{S}_{n})\Pi_{\lambda}\}
\end{align*}
where $\Pi_\lambda$ is the projector to the subspace of $(\mathbb{C}^{d})^{\otimes n}$ which carries irrep $\lambda$ of $\mathcal{S}_n$. From theorem \ref{Thm-sym} we know that for the special case of 1-d representations of $\mathcal{S}_n$, i.e. for symmetric and anti-symmetric subspaces, the above equality hold.   Here, we  build an explicit counter-example to this equality for other irreps of $\mathcal{S}_n$.

 First, notice that the   action of $\textbf{Q}(G)$, $\textbf{Q}(G')$ and $\textbf{P}(\mathcal{S}_{n})$ on $(\mathbb{C}^{d})^{\otimes n}$ all commute with each other. This implies that for irrep $\lambda$ of $\mathcal{S}_{n}$ the subspace $\left[( {\mathbb{C}^{d} })^{\otimes n}\right]_{\lambda}$ can be decomposed as
\begin{equation} \label{count-ex}
\left[({\mathbb{C}^{d}})^{\otimes n} \right]_{\lambda}\cong\left(\bigoplus_{\mu,\nu} \mathcal{M}_{\mu}\otimes\mathcal{N}_{\nu}\otimes  \mathbb{C}^{m_{\mu,\nu}} \right)\otimes\mathcal{K}_{\lambda}
\end{equation}
where  $\mu$ labels irreps of $G$ and $\nu$ labels irreps of $G'$  and  furthermore $\textbf{Q}(G)$, $\textbf{Q}(G')$ and  $\textbf{P}(\mathcal{S}_{n})$ act nontrivially only on  $\mathcal{M}_{\mu}$, $\mathcal{N}_{\nu}$ and $\mathcal{K}_{\lambda}$ respectively. Note that any permutationally invariant operator should be proportional to identity on the subsystem $\mathcal{K}_{\lambda}$.

Now to build the counterexample we find two gauge groups  $G$ and $G'$ for which there is no one-to-one relation between the irreps of $G$ and $G'$ which show up in   $\left[({\mathbb{C}^{d}})^{\otimes n} \right]_{\lambda}$. In other words, we find an  example in which there is some irrep $\mu$ of $G$ for which $m_{\mu,\nu}$ is nonzero for more than one $\nu$ (irrep of $G'$).  This in turn will imply that there exist  permutationally invariant operators  $\Pi_{\lambda}M\Pi_{\lambda}$ which commute with  $\Pi_{\lambda}\textbf{Q}(G)\Pi_{\lambda}$ but are not block diagonal between  $\mathcal{N}_{\nu_{1}}$ and $\mathcal{N}_{\nu_{2}}$ for two different irrep $\nu_{1}$ and $\nu_{2}$ of $G'$. This implies that   $\Pi_{\lambda}M\Pi_{\lambda}$ cannot be in $\text{Alg}\{\Pi_{\lambda}\textbf{Q}(G')\Pi_{\lambda}\}$.



Note that from theorem  \ref{Thm-sym}  we know that in the specific case where irrep $\lambda$ is a 1-d representations of $\mathcal{S}_{n}$ the conjecture holds.  So to find a counter-example we need to look at $n>2$ where the permutation group can have irreps  other than the symmetric and anti-symmetric. In the following, we present a counter-example in the case of $n=3$. In this case the permutation group $\mathcal{S}_{3}$ has a two dimensional irrep denoted by $\lambda_{2}$.

\subsubsection{{Counter-example}}

Consider the Hilbert space $\mathbb{C}^{4} \cong \mathcal{H}_{L}\otimes \mathcal{H}_{R}$  where $\mathcal{H}_{L}$ and $\mathcal{H}_{R}$ are both isomorphic to $\mathbb{C}^{2}$ . Suppose $G=\{V\otimes I: V\in U(2)\}$, i.e. the group of all unitaries  which act trivially on $\mathcal{H}_{R}$. Clearly $G'$ is the set of all unitaries acting trivially on  $\mathcal{H}_{L}$ and so $G=(G')'$.  Note that both $G$ and $G'$ are isomorphic to $\text{U}(2)$. So in decomposition \ref{count-ex} we can label irreps of   $G$ and $G'$ with irreps of $\text{U}(2)$.

Using decomposition  $(\mathbb{C}^4)^{\otimes 3}\cong\mathcal{H}_{L}^{\otimes 3}\otimes \mathcal{H}_{R}^{\otimes 3} $ we can think of the collective representation of $G$ and $G'$  on $(\mathbb{C}^4)^{\otimes 3}$ as
$$V\otimes \mathbb{I}_{R}\in G\rightarrow \textbf{Q}(V\otimes \mathbb{I}_{R})=\textbf{Q}_{L}(V)\otimes \mathbb{I}^{\otimes 3}_{R}$$
and
$$ \mathbb{I}_{L} \otimes V \in G'\rightarrow  \textbf{Q}(\mathbb{I}_{L}\otimes V)= \mathbb{I}^{\otimes 3}_{L} \otimes \textbf{Q}_{R}(V)$$
respectively where $V\rightarrow \textbf{Q}_{L/R}(V)=V^{\otimes 3}$  can be thought as the collective representation of $\text{U}(2)$ on  $\mathcal{H}_{L/R}^{\otimes 3}$, $\mathbb{I}_{L/R}$ is the identity operator on $\mathcal{H}_{L/R}$ and so $\mathbb{I}^{\otimes 3}_{L/R}$ is the identity operator on $\mathcal{H}^{\otimes 3}_{L/R}$ .

Similarly we can think of the canonical representation of  $\mathcal{S}_{3}$ on $(\mathbb{C}^4)^{\otimes 3}$ as
$$\textbf{P}(s\in \mathcal{S}_{3})=\textbf{P}_{L}(s)\otimes \textbf{P}_{R}(s)$$
where $\textbf{P}_{L}(\mathcal{S}_{3})$ and $\textbf{P}_{R}(\mathcal{S}_{3})$ are the canonical representation of $\mathcal{S}_{3}$ on  $\mathcal{H}_{L}^{\otimes 3}$ and $\mathcal{H}_{R}^{\otimes 3}$ respectively.

Now according to Schur-Weyl duality there is a one to one relation between the irreps of $\text{U}(2)$ which show up in representation $\textbf{Q}_{L/R}(\text{U}(2))$ on   $(\mathcal{H}_{L/R})^{\otimes 3}$ and irreps of $\mathcal{S}_{3}$ which show up in representation $\textbf{P}_{L/R}(\mathcal{S}_{3})$ on   $(\mathcal{H}_{L/R})^{\otimes 3}$. Note that under the action of $\mathcal{S}_{3}$, $\mathcal{H}_{L/R}^{\otimes 3}$ decomposes as
$$\mathcal{H}_{L/R}^{\otimes 3}\cong   \left[\mathcal{H}_{L/R}^{\otimes 3}\right]_{{+}} \oplus \left[\mathcal{H}_{L/R}^{\otimes 3}\right]_{{\lambda_{2}}}$$
(the anti-symmetric irrep of $\mathcal{S}_{3}$ does not exist in this representation.)
Now Schur-Weyl duality implies that in the representation $\textbf{Q}_{L/R}(\text{U}(2))$  of $\text{U}(2)$ only one irrep of $\text{U}(2)$  shows up in the subspace    $\left[\mathcal{H}_{L/R}^{\otimes 3}\right]_{{+}}$ and a different one will show up in  $\left[\mathcal{H}_{L/R}^{\otimes 3}\right]_{{\lambda_{2}}}$.

This implies that there is a one-to-one relation between irreps of $\text{U}(2)$ which show up in representation $\textbf{Q}_{L}(\text{U}(2))\otimes I_{R}$ in the total Hilbert space $(\mathbb{C}^{4})^{\otimes 3}$  and irreps of $\mathcal{S}_{3}$ which show up in the representation $\textbf{P}_{L}(\mathcal{S}_{3})\otimes I_{R}$ in the total Hilbert space $(\mathbb{C}^{4})^{\otimes 3}$ (though $(\textbf{P}_{L}(\mathcal{S}_{3})\otimes I_{R}) \times (\textbf{Q}_{L}(\text{U}(2))\otimes I_{R}) $ is no longer multiplicity-free).  In other words, in representation $\textbf{Q}_{L}(\text{U}(2))\otimes I_{R}$ only one irrep of $\text{U}(2)$ shows up in the subspace

$$\left[\mathcal{H}_{L}^{\otimes 3}\right]_{{+}} \otimes \left( \left[\mathcal{H}_{R}^{\otimes 3}\right]_{{+}} \oplus \left[\mathcal{H}_{R}^{\otimes 3}\right]_{{\lambda_{2}}}  \right) $$
and a different one shows up in
$$\left[\mathcal{H}_{L}^{\otimes 3}\right]_{{\lambda_{2}}} \otimes \left( \left[\mathcal{H}_{R}^{\otimes 3}\right]_{{+}} \oplus \left[\mathcal{H}_{R}^{\otimes 3}\right]_{{\lambda_{2}}}  \right) $$

Similarly, under $I_{L}\otimes \textbf{Q}_{R}(\text{U}(2))$ only one irrep of $\text{U}(2)$ shows up in the subspace

$$\left( \left[\mathcal{H}_{L}^{\otimes 3}\right]_{{+}} \oplus \left[\mathcal{H}_{L}^{\otimes 3}\right]_{{\lambda_{2}}}  \right) \otimes \left[\mathcal{H}_{R}^{\otimes 3}\right]_{{+}} $$
and a different one shows up in
$$\left( \left[\mathcal{H}_{L}^{\otimes 3}\right]_{{+}} \oplus \left[\mathcal{H}_{L}^{\otimes 3}\right]_{{\lambda_{2}}}  \right) \otimes \left[\mathcal{H}_{R}^{\otimes 3}\right]_{{\lambda_{2}}} $$


Now we find which parts of these subspaces of  $(\mathbb{C}^{4})^{\otimes 3}$ form $\left[(\mathbb{C}^{4})^{\otimes 3}\right]_{\lambda_{2}}$ and we show that in this subspace there is no one-to-one relation between irreps of $\text{U}(2)$ which show up in the representation  $\textbf{Q}_{L}(\text{U}(2))\otimes I_{R}$ and irreps of $\text{U}(2)$  which show up in the representation  $I_{L}\otimes \textbf{Q}_{R}(\text{U}(2))$.

To see this consider the total Hilbert space
\begin{align*}
(\mathbb{C}^{4})^{\otimes 3} \cong &\ \\
&\left(\left[\mathcal{H}_{L}^{\otimes 3}\right]_{{+}} \otimes \left[\mathcal{H}_{R}^{\otimes 3}\right]_{{+}} \right) \\
  \oplus&\left(\left[\mathcal{H}_{L}^{\otimes 3}\right]_{{\lambda_{2}}} \otimes \left[\mathcal{H}_{R}^{\otimes 3}\right]_{{+}}\right) \oplus \left(\left[\mathcal{H}_{L}^{\otimes 3}\right]_{{+}} \otimes  \left[\mathcal{H}_{R}^{\otimes 3}\right]_{{\lambda_{2}}}\right)\\
 \oplus&\left(\left[\mathcal{H}_{L}^{\otimes 3}\right]_{{\lambda_{2}}} \otimes \left[\mathcal{H}_{R}^{\otimes 3}\right]_{{\lambda_{2}}}\right)
\end{align*}
and $\textbf{P}(\mathcal{S}_{n})$ the canonical representation of $\mathcal{S}_{3}$ on it. Then,  $\textbf{P}(s\in\mathcal{S}_{n})=\textbf{P}_{L}(s)\otimes \textbf{P}_{R}(s)$ implies that: i) the subspace in the first line is in the symmetric subspace of  $(\mathbb{C}^{4})^{\otimes 3}$, i.e. in $\left[(\mathbb{C}^{4})^{\otimes 3}\right]_{+}$ (and so we do not care about it), ii) the subspace in the second line is in $\left[(\mathbb{C}^{4})^{\otimes 3}\right]_{\lambda_{2}}$ and iii) a nonzero subspace of the subspace in the third line is also  in $\left[(\mathbb{C}^{4})^{\otimes 3}\right]_{\lambda_{2}}$. To see this  we note that the action of $\mathcal{S}_{3}$ on  $ \left[\mathcal{H}_{L}^{\otimes 3}\right]_{{\lambda_{2}}} \otimes \left[\mathcal{H}_{R}^{\otimes 3}\right]_{{\lambda_{2}}} $ is non-commutative and since the only irrep of $\mathcal{S}_{3}$ in which the representation of $\mathcal{S}_{3}$ is non-commutative is $\lambda_{2}$, therefore by decomposing the action of   $\textbf{P}(\mathcal{S}_{n})$  on $ \left[\mathcal{H}_{L}^{\otimes 3}\right]_{{\lambda_{2}}} \otimes \left[\mathcal{H}_{R}^{\otimes 3}\right]_{{\lambda_{2}}} $ to irreps one should find a $\lambda_{2}$ irrep.

This implies that in $\left[(\mathbb{C}^{4})^{\otimes 3}\right]_{\lambda_{2}}$, $\left[\mathcal{H}_{L}^{\otimes 3}\right]_{{\lambda_{2}}}$ couples to both $\left[\mathcal{H}_{R}^{\otimes 3}\right]_{{+}}$ and $\left[\mathcal{H}_{R}^{\otimes 3}\right]_{{\lambda_{2}}}$.  This in turn will imply that there is no one-to-one relation between the irreps of $\text{U}(2)$ which show up in representations $\textbf{Q}_{L}(\text{U}(2))\otimes I_{R}$ and $I_{L}\otimes \textbf{Q}_{R}(\text{U}(2))$ in the subspace $\left[(\mathbb{C}^{4})^{\otimes 3}\right]_{\lambda_{2}}$.

Therefore, this example is a counter-example to the above conjecture.

\section{Cost function} \label{sec:costfunctions}
Here, we present the average cost function as an example  of a common figure of merit and we show that it can be accommodated within the framework we introduced in section \ref{sec:estimation}.


Suppose that ${\mathfrak{s}}(\rho)$ is a parameter to be estimated.  As described earlier, any estimation scheme, consisting of a choice of measurement and a post-processing of its outcome, can be described by a POVM $M:\sigma(\Omega)\rightarrow\text{End}((\mathbb{C}^d)^{\otimes n})$.  In this case, the outcome space $\Omega$ must correspond to the range of ${\mathfrak{s}}$. In the following we use the differential notation $M(\textrm{d}{S_{\text{est}}})$ to show POVM so that for any interval $\Delta\subseteq \mathbb{R}$
\begin{equation}
M({\Delta})=\int_{{\Delta}} M(\textrm{d}{S_{\text{est}}})
\end{equation}
Therefore, using the strategy $M$ the conditional probability of outcomes for state $\rho$ will be 
\begin{equation}
q_M(\textrm{d}{S_{\text{est}}}|\rho) = \textrm{tr}(M(\textrm{d}{S_{\text{est}}})\rho).
\end{equation}
Now suppose that the performance of the estimation scheme will be judged by a cost function (we follow Ref.~\cite{Chiribella}).  The most basic case would be a function of the form $C({S}_{\text{est}}, {\mathfrak{s}}(\rho))$, which represents the cost of estimating ${S}_{\text{est}}$ when the true value of the parameters is ${\mathfrak{s}}(\rho)$.   

The average cost of the estimation strategy $M$ for the state $\rho$ is
\begin{equation} \label{eq:jj1}
\overline{C}_{M}(\rho) \equiv \int  C({S}_{\text{est}}, {\mathfrak{s}}(\rho))\ \  q_{M}( dS_{\text{est}}|\rho) \; 
\end{equation}
and the expected cost of the estimation strategy $M$ given the prior density $p$ is
\begin{equation}\label{eq:jj2}
\langle C \rangle_{M} \equiv \int d\rho\; p(\rho)\ \overline{C}_{M}(\rho).
\end{equation}
Therefore, 
\begin{align*}
\langle C \rangle_{M}&= 
\int d\rho\ p(\rho) \int  C({S}_{\text{est}}, {\mathfrak{s}}(\rho))\ \ q_{M}( dS_{\text{est}}|\rho)=\int\int p(S)\ C({S}_{\text{est}}, S)\ \  q_{M}(dS_{\text{est}}|dS)
\end{align*}
where $S$ is the random variable defined by  the function $\mathfrak{s}$ acting on the random state $\rho$ and $p(S)$ is the density of random variable $S$ relative to $dS$.
So this figure of merit is clearly a functional of $q_{M}(dS_{\text{est}}|dS)$ and hence the condition of corollary~\ref{corollary:main} applies.  It follows that if the problem has gauge symmetry $G_{\mathcal{A}}$ and satisfies the assumptions of theorem~\ref{Thm-main-generalized} (or theorem~\ref{Thm-main-chann}), then the optimal estimation can be achieved with POVMs restricted to $\mathcal{A}^{\otimes n}$.

\end{document}